\documentclass[sigconf]{acmart}
\newcommand{\myparatight}[1]{\smallskip\noindent{\bf {#1}:}~}

\usepackage{amsthm}
\usepackage{amsmath}
 
\usepackage{amssymb}
\usepackage{float}
\usepackage{graphics}
\usepackage{graphicx}
\usepackage[skip=0pt]{caption}
\usepackage{algorithm}
\usepackage{algpseudocode}
\usepackage{bm}
\usepackage{color}
\usepackage{multirow}
\usepackage{makecell}
\usepackage{subfig}
\usepackage{soul}
\usepackage{xcolor}

\usepackage{subfig}
\usepackage{wrapfig}

\usepackage{appendix}

\usepackage{graphics}
\usepackage{graphicx}

\allowdisplaybreaks
\usepackage{xspace}

\usepackage{hyperref}
\hypersetup{
	colorlinks=true,
	linkcolor=blue,
	filecolor=magenta,      
	urlcolor=cyan,
        citecolor=blue,
	pdftitle={Overleaf Example},
	pdfpagemode=FullScreen,
}

\usepackage[scaled]{beramono}

\usepackage{xspace}

\newcommand{\alg}{{\textsf{SecureAFL}}\xspace}

\usepackage{color, colortbl}

\definecolor{greyL}{RGB}{230,248,255}

\usepackage{balance}

\newtheorem{assumption}{Assumption}
\newtheorem{thm}{Theorem}
\newtheorem{lem}{Lemma}
\newtheorem*{remark}{Remark}

\algnewcommand\algorithmicforpara{\textbf{for}}
\algnewcommand\algorithmicdoinparallel{\textbf{do in parallel}}
\algdef{S}[FOR]{ForParallel}[1]{\algorithmicforpara\ #1\ \algorithmicdoinparallel}

\copyrightyear{2026}
\acmYear{2026}
\setcopyright{cc}
\setcctype{by}
\acmConference[ASIA CCS '26]{ACM Asia Conference on Computer and Communications Security}{June 01--05, 2026}{Bangalore, India}
\acmBooktitle{ACM Asia Conference on Computer and Communications Security (ASIA CCS '26), June 01--05, 2026, Bangalore, India}
\acmDOI{10.1145/3779208.3807486}
\acmISBN{979-8-4007-2356-8/2026/06}

\begin{document}

\captionsetup[subfigure]{skip=0pt} 
\captionsetup[subtable]{skip=0pt}

\title{SecureAFL: Secure Asynchronous Federated Learning}

\author{Anjun Gao}
\authornote{Equal contribution.}
\affiliation{
	\institution{University of Louisville}
	\city{Louisville}
        \country{USA}
}

\author{Feng Wang}
\authornotemark[1]
\authornote{Feng Wang conducted this research while he was an intern under the supervision of Minghong Fang.}
\affiliation{
	\institution{Northeastern University}
	\city{Shenyang}
        \country{China}
}

\author{Zhenglin Wan}
\affiliation{
	\institution{National University of Singapore}
	\city{Singapore}
        \country{Singapore}
}

\author{Yueyang Quan}
\affiliation{
	\institution{University of North Texas}
	\city{Denton}
        \country{USA}
}

\author{Zhuqing Liu}
\affiliation{
	\institution{University of North Texas}
	\city{Denton}
        \country{USA}
}

\author{Minghong Fang}
\affiliation{
	\institution{University of Louisville}
	\city{Louisville}
         \country{USA}
}

\renewcommand{\shortauthors}{Anjun Gao et al.}

\begin{abstract}
Federated learning (FL) enables multiple clients to collaboratively train a global machine learning model via a server without sharing their private training data. In traditional FL, the system follows a synchronous approach, where the server waits for model updates from numerous clients before aggregating them to update the global model. However, synchronous FL is hindered by the straggler problem. To address this, the asynchronous FL architecture allows the server to update the global model immediately upon receiving any client’s local model update. Despite its advantages, the decentralized nature of asynchronous FL makes it vulnerable to poisoning attacks. Several defenses tailored for asynchronous FL have been proposed, but these mechanisms remain susceptible to advanced attacks or rely on unrealistic server assumptions. In this paper, we introduce \alg, an innovative framework designed to secure asynchronous FL against poisoning attacks. \alg improves the robustness of asynchronous FL by detecting and discarding anomalous updates while estimating the contributions of missing clients. Additionally, it utilizes Byzantine-robust aggregation techniques, such as coordinate-wise median, to integrate the received and estimated updates. Extensive experiments on various real-world datasets demonstrate the effectiveness of \alg.
\end{abstract}

\ccsdesc[500]{Security and privacy~Systems security}

\keywords{Asynchronous Federated learning, Poisoning Attacks, Robustness}

\maketitle


\section{Introduction} 
\label{sec:intro}

Federated learning (FL)~\cite{mcmahan2017communication} is a distributed machine learning paradigm that has gained considerable traction in recent years. It allows multiple clients to collaboratively train a shared global model under the coordination of a central server. During each training round, the server transmits the global model to the clients, who then refine their local models using their respective datasets. The clients subsequently send their local model updates back to the server, which aggregates these updates to enhance the global model. With its emphasis on respecting client privacy, FL has been widely adopted across diverse domains~\cite{webank,gboard,paulik2021federated}.

Most existing FL systems naturally adopt a {\em synchronous} design~\cite{blanchard2017machine,mhamdi2018hidden,munoz2019byzantine,yin2018byzantine,cao2020fltrust,dou2025toward,fung2020limitations,fang2025byzantine,wang2025poisoning}, where the server waits to receive local model updates from all or most clients before performing aggregation. While this simplifies FL by ignoring global model discrepancies among clients, it suffers from the {\em straggler problem}, where clients with slower hardware or network delays take significantly longer to transmit updates. 
For instance,
training large-scale models, such as GPT~\cite{brown2020language}, is computationally intensive, often requiring days or even weeks on high-performance GPUs. Clients with limited computational resources may take even longer, leading to severe delays in FL. As model sizes continue to grow, the time required for local training increases accordingly, exacerbating the inefficiency of synchronous FL. 
A simple solution is to discard updates from slow clients, but this can harm model accuracy and waste resources~\cite{tandon2017gradient}.

The limitations of synchronous FL underscore the need for an {\em asynchronous} design. In asynchronous FL~\cite{nguyen2022federated, huba2022papaya,chen2020asynchronous,xu2021asynchronous,xie2019asynchronous,chen2020vafl,van2020asynchronous,wang2022asynchronous,liu2024fedasmu}, clients operate on different versions of the global model when refining their local models, and, crucially, the server updates the global model as soon as it receives an update from any client, rather than waiting for all clients to complete their training. This enables continuous learning without being bottlenecked by slower clients. However, since clients may use outdated versions of the global model, update staleness can arise. Despite this challenge, asynchronous FL offers significant advantages in handling straggling clients, improving resource utilization, and enhancing scalability. These benefits have led to its widespread adoption in various applications~\cite{abadi2016tensorflow,paszke2019pytorch,nguyen2022federated, huba2022papaya}, particularly in scenarios where real-time updates are essential.

Like its synchronous counterpart, asynchronous FL is also susceptible to poisoning attacks~\cite{tolpegin2020data,fang2020local,blanchard2017machine,bagdasaryan2020backdoor,xie2019dba,sun2019can,zhang2022neurotoxin,li20233dfed,shejwalkar2021manipulating,zhang2024poisoning,yin2024poisoning,xie2024poisonedfl}. Malicious clients controlled by an attacker can corrupt their local training data or alter their model updates before sending them to the server. This allows the attacker to manipulate the global model to serve its objective, such as causing widespread misclassification~\cite{tolpegin2020data,fang2020local,blanchard2017machine,shejwalkar2021manipulating} or targeting specific predictions~\cite{bagdasaryan2020backdoor,xie2019dba,sun2019can,zhang2022neurotoxin,li20233dfed}.  
To mitigate poisoning attacks, various Byzantine-robust aggregation rules have been developed. However, most of these defenses are designed for synchronous FL~\cite{cao2020fltrust,nguyen2022flame,pan2020justinian,park2021sageflow,wang2022flare,xie2019zeno,ChenPOMACS17,kumari2023baybfed,mhamdi2018hidden,fang2023robust}, where the server can analyze statistical patterns across multiple client updates received simultaneously. These synchronous-based defenses, however, are not applicable to asynchronous FL, as updates arrive one at a time, making statistical anomaly detection unfeasible.  
Another key challenge in securing asynchronous FL lies in handling delayed updates, which often introduce noise and further complicate the server’s ability to differentiate between benign and malicious updates. In response to these issues, a few robust asynchronous FL frameworks~\cite{fang2022aflguard,damaskinos2018asynchronous,yang2021basgd,xie2020zeno++} have emerged in recent years. However, existing approaches either remain vulnerable to poisoning attacks~\cite{damaskinos2018asynchronous,yang2021basgd} or rely on strong assumptions about the server’s capabilities~\cite{fang2022aflguard,xie2020zeno++}, such as assuming access to a separate trusted dataset, an assumption that rarely holds in practice.

In this paper, we introduce \alg, a defense mechanism designed for asynchronous FL that enhances robustness by identifying and discarding anomalous updates while estimating missing client contributions.
Our \alg first applies a filtering mechanism based on the Lipschitz continuity of local updates, ensuring that only those conforming to historical patterns are accepted. By tracking the evolution of client updates over time, the server evaluates the smoothness of updates and discards those that exhibit abrupt deviations, which could indicate adversarial manipulation. This approach enables the server to systematically mitigate the impact of malicious updates while preserving the integrity of benign contributions.

Beyond filtering, our proposed \alg introduces a local update estimation strategy to reconstruct missing client updates using historical information. By approximating these missing updates based on past update trajectories, the server infers plausible updates for clients that have not yet contributed in the current round. Finally, a Byzantine-robust aggregation mechanism, such as the coordinate-wise median~\cite{yin2018byzantine}, is employed to integrate the received and estimated updates, further improving robustness against malicious updates. By combining filtering, estimation, and robust aggregation, our \alg enhances the security and stability of asynchronous FL systems, effectively mitigating the risks posed by poisoning attacks.

We thoroughly assess the performance of our \alg using five datasets from various domains, including large-scale benchmark datasets such as CIFAR-10~\cite{cifar10data} and CIFAR-100~\cite{cifar10data}, as well as the real-world autonomous driving dataset Udacity~\cite{Udacity}. Our evaluation involves ten poisoning attacks consisting of five untargeted attacks, five targeted attacks, and a particularly challenging adaptive attack, alongside comparisons with seven recent asynchronous FL approaches. Experimental results demonstrate that \alg effectively mitigates diverse existing and adaptive poisoning attacks, even when facing scenarios with a high proportion of malicious clients. Furthermore, \alg achieves significant improvements compared to existing asynchronous FL methods designed to be robust against Byzantine adversaries.
The key contributions of this paper are summarized as follows:

\begin{list}{\labelitemi}{\leftmargin=1.15em \itemindent=-0.08em \itemsep=.1em}

\item
We introduce \alg, a novel defense framework designed to mitigate poisoning attacks in asynchronous federated learning.

\item
By conducting a thorough evaluation against 10 different poisoning attacks across 5 datasets and benchmarking \alg against 7 asynchronous federated learning methods, we validate its effectiveness in mitigating the effects of poisoning attacks.

\item
We demonstrate the resilience of \alg against powerful adaptive attacks and highlight its sustained effectiveness even when a substantial portion of clients exhibit malicious behavior.

\end{list}


\begin{table}[t]
	\addtolength{\tabcolsep}{-2.5pt}
	\caption{Summary of key notation.}
	\centering
	\begin{tabular}{|c|c|} \hline 
		{Notation} & {Definition} \\ \hline
		$n$ & Number of clients \\ \hline
		$\bm{w}^t$ & Global model at the \( t \)th training round \\ \hline
        $\bm{g}_i^t$ & Local model update from client $i$ at round $t$ \\ \hline
        $\hat{\bm{g}}_i^t$ & Estimated local model update from client $i$ at round $t$ \\ \hline
    	$\tau_i$ & Delay of client $i$'s local model update \\ \hline
    	$\lambda_i^t$ & Lipschitz factor for client \( i \) at round \( t \) \\ \hline
	\end{tabular}
	\label{main_notaion}
       \vspace{-.05in}
\end{table}

\section{Preliminaries and Related Work}
\label{sec:preliminaries}

\myparatight{Notations}%
In this paper, \(\left\| \cdot \right\|\) represents the \(\ell_2\)-norm, while \([n]\) denotes the set \(\{1, 2, \dots, n\}\).
Table~\ref{main_notaion} lists the key notation used in this paper.

\subsection{Asynchronous FL} 
\label{sec:background}

\begin{table*}[t]
  \centering
    \addtolength{\tabcolsep}{-0.05pt}
  \caption{Comparison of Synchronous FL, Semi-asynchronous FL, and Asynchronous FL.}
  \label{tab:compare_FL}
  \newcolumntype{C}[1]{>{\centering\arraybackslash}m{#1}}

  \begin{tabular}{|C{3.6cm}|C{6cm}|C{4cm}|C{2.5cm}|}
    \hline
     & Representative methods & Batch update or single update & Delay or no delay \\
    \hline
    Synchronous FL &
    Median~\cite{yin2018byzantine}, Trimmed-mean~\cite{yin2018byzantine}, 
    Krum~\cite{blanchard2017machine}, FLAME~\cite{nguyen2022flame}, 
    Baybfed~\cite{kumari2023baybfed}, 
    Bucketing~\cite{karimireddy2020byzantine}, Foolsgold~\cite{fung2020limitations}, 
    DnC~\cite{shejwalkar2021manipulating}, RFLBAT~\cite{wang2022rflbat}, 
    SignGuard~\cite{xu2022byzantine}, FLShield~\cite{kabir2024flshield}, 
    FreqFed~\cite{fereidooni2023freqfed}, 
    BackdoorIndicator~\cite{li2024backdoorindicator}, 
    FLTrust~\cite{cao2020fltrust}, GAA~\cite{pan2020justinian}, 
    FLARE~\cite{wang2022flare},
    FoundationFL~\cite{fang2025we}
    & Batch update & No delay \\
    \hline
    Semi-asynchronous FL &
    Catalyst~\cite{cox2024asynchronous}, PoiSAFL~\cite{pang2025poisafl}
    & Batch update & Delay \\
    \hline
    Asynchronous FL &
    Kardam~\cite{damaskinos2018asynchronous}, BASGD~\cite{yang2021basgd}, 
    Sageflow~\cite{park2021sageflow}, Zeno++~\cite{xie2020zeno++}, 
    AFLGuard~\cite{fang2022aflguard}, AsyncDefender~\cite{bai2025asyncdefender}, \alg (Ours)
    & Single update & Delay \\
    \hline
  \end{tabular}
\end{table*}

In a federated learning (FL) system comprising \( n \) clients, each client \( i \) maintains a local dataset \( D_i \) for \( i \in [n] \).
For convenience, let $ D = \bigcup_{i \in [n]} D_i $ denote the combined training dataset across all clients and \( \mathcal{L} \) represents the loss function. The clients collaboratively train a global model by minimizing a global objective function defined as
$
\min_{\bm{w}} F(\bm{w})
=
\min_{\bm{w}} \sum_{i \in [n]} f_i(\bm{w}),
$
where \( \bm{w} \) denotes the model parameters, \( f_i(\bm{w}) = \mathcal{L}(D_i;\bm{w}) \) is the local objective of client \( i \).

Traditional FL follows a synchronous approach, where the server waits for all client updates before aggregation. However, this process is often slowed by stragglers—clients with delayed submissions.
To mitigate this issue, asynchronous FL has been introduced. Unlike in synchronous FL, where all clients receive the same global model at the beginning of a round and update it simultaneously, asynchronous FL allows the server to update the global model immediately upon receiving a local update from any client~\cite{hard2024learning,xu2021asynchronous,xie2019asynchronous,chen2020vafl}. As a result, clients operate independently and may fetch the global model at different times, leading to variations in the model versions they use for local training. Consequently, each client fine-tunes its local model using a potentially outdated version of the global model, introducing update staleness. 
Specifically, let \( \bm{w}^t \) represent the global model at the \( t \)th training round, and let \( \bm{g}_i^t \) denote the local model update from client \( i \), which is computed based on \( \bm{w}^{t} \). In the \( t \)th round, suppose the server receives a model update \( \bm{g}_i^{t-\tau_i} \) from client \( i \), which was computed using an earlier global model \( \bm{w}^{t-\tau_i} \) from round \( t-\tau_i \), where \( \tau_i \) represents the delay in the local model update. Upon receiving this update, the server incorporates it into the global model through the following update rule:
\begin{align}
\label{updaterule}
\bm{w}^{t+1} =  \bm{w}^{t} -\eta \bm{g}_i^{t-\tau_i},
\end{align}
where $\eta$ represents the global learning rate.

The procedure for asynchronous FL is outlined in Algorithm~\ref{AsynSGD_alg}.
In this algorithm, $T$ represents the total number of training rounds.

\begin{algorithm}[t]
	\caption{Asynchronous FL.}
	\label{AsynSGD_alg}
	\begin{algorithmic}[1]
		 \Statex  \underline{Server:} 
		\State Initializes the global model \(\bm{w}^{0}\) and distributes it to all clients.
		\For {$t=0,1,2,\cdots,T-1$} 
		\State Upon receiving model update $\bm{g}_i^{t-\tau_i}$ from client $i$, the global model is updated as $\bm{w}^{t+1} = \bm{w}^{t} - \eta \bm{g}_i^{t-\tau_i}$.
        \State Sends the updated global model $\bm{w}^{t+1} $ to client $i$.
		\EndFor
		\Statex \underline{Client $i$, $i \in [n]$:} 
		\Repeat
			\State Obtains the global model $\bm{w}^t$ transmitted by the server.
			\State Calculates the gradient $\bm{g}_i^t$ using $\bm{w}^t$ and the local training dataset, then transmits $\bm{g}_i^t$ to the server.
 		\Until{Convergence}
	\end{algorithmic} 
\end{algorithm}

\subsection{Poisoning Attacks to FL} 
\label{sec:Poisoning_Attacks_FL}

The decentralized nature of FL makes the global model vulnerable to poisoning attacks through poisoned local training data or altered local model updates. Poisoning attacks in FL can be categorized as untargeted~\cite{tolpegin2020data,fang2020local,blanchard2017machine,shejwalkar2021manipulating} or targeted~\cite{bagdasaryan2020backdoor,xie2019dba,sun2019can,zhang2022neurotoxin,li20233dfed}.  
Untargeted attacks, such as label flipping attack~\cite{tolpegin2020data}, sign flipping attack~\cite{fang2020local}, and Gaussian attack~\cite{blanchard2017machine}, aim to degrade overall model performance by introducing misleading updates. 
Targeted attacks, such as backdoor attacks~\cite{bagdasaryan2020backdoor,xie2019dba,sun2019can,zhang2022neurotoxin,li20233dfed}, manipulate the model to misbehave only under specific conditions. For example, in the distributed backdoor attack (DBA) attack~\cite{xie2019dba}, the attacker embeds different backdoor triggers into the training data of malicious clients, ensuring the model behaves normally in most cases but produces malicious outputs for specific inputs.

\subsection{Robust Aggregation in Synchronous FL}

Most existing robust aggregation rules for federated learning, including Median and Trimmed-mean~\cite{yin2018byzantine}, Krum~\cite{blanchard2017machine}, FLAME~\cite{nguyen2022flame},
Baybfed~\cite{kumari2023baybfed}, as well as many others \cite{cao2020fltrust,pan2020justinian,wang2022flare,xie2019zeno,ChenPOMACS17,mhamdi2018hidden,fereidooni2023freqfed,zhang2022fldetector,mozaffari2023every,guerraoui2018hidden,cao2021provably,li2019rsa,pillutla2022robust,fang2024byzantine,fang2025we,fang2025provably,xu2024robust,mo2025find}, are designed for synchronous FL settings.
In synchronous FL, the server collects multiple client updates computed from the same global model before aggregation. This batch-based setting enables statistical comparisons across updates, allowing the server to identify and mitigate malicious behavior using robust aggregation rules.

\subsection{Existing Defenses for Asynchronous FL}

In recent years, several Byzantine-robust methods have been proposed specifically for asynchronous FL \cite{fang2022aflguard,damaskinos2018asynchronous,yang2021basgd,xie2020zeno++}.
For example, BASGD \cite{yang2021basgd} introduces a buffering mechanism that accumulates delayed updates, averages updates within each buffer, and then applies a median-based aggregation across buffers. Other approaches, such as Sageflow \cite{park2021sageflow}, Zeno++ \cite{xie2020zeno++}, and AFLGuard \cite{fang2022aflguard}, rely on the availability of a trusted dataset at the server to generate reference updates and evaluate incoming updates based on their alignment with this reference.
While these methods adapt robustness techniques to asynchronous settings, they rely on additional assumptions or mechanisms that may not hold in practice.
A more recent defense, AsyncDefender~\cite{bai2025asyncdefender}, assigns weights according to the cosine similarity between client updates and the global model.

\subsection{Synchronous vs. Semi-Asynchronous vs. Asynchronous FL}

Table~\ref{tab:compare_FL} summarizes the key differences between synchronous FL, semi-asynchronous FL, and asynchronous FL.
Synchronous FL adopts a batch-update mechanism without update delay, enabling a wide range of robust aggregation rules but suffering from the straggler problem. Semi-asynchronous FL relaxes strict synchronization by allowing delayed updates while still performing batch aggregation, which reduces straggler impact and permits limited robust defenses such as Catalyst~\cite{cox2024asynchronous} and PoiSAFL~\cite{pang2025poisafl}.
In contrast, asynchronous FL updates the global model immediately upon receiving a single client update, eliminating server-side waiting and improving scalability. However, this single-update and delayed-update setting prevents the direct application of batch-based robust aggregation rules, requiring fundamentally different defense designs such as Kardam~\cite{damaskinos2018asynchronous}, BASGD~\cite{yang2021basgd}, Sageflow~\cite{park2021sageflow}, Zeno++~\cite{xie2020zeno++}, AFLGuard~\cite{fang2022aflguard}, 
AsyncDefender~\cite{bai2025asyncdefender}
and our proposed \alg.

\subsection{Limitations of Existing Defenses}

Despite recent progress, existing defenses for asynchronous FL exhibit notable limitations. First, methods such as Kardam~\cite{damaskinos2018asynchronous} and BASGD~\cite{yang2021basgd} fail to effectively mitigate poisoning attacks under strong adversarial settings, as demonstrated in our experiments. Second, defenses including Sageflow~\cite{park2021sageflow}, Zeno++\cite{xie2020zeno++}, and AFLGuard~\cite{fang2022aflguard} impose strong assumptions on the server, requiring access to a trusted dataset that closely matches the clients’ data distribution, an assumption that is often unrealistic.
AsyncDefender~\cite{bai2025asyncdefender} weights updates by their directional alignment with the global model, but this design can be exploited by the attacker who crafts well-aligned malicious updates to evade detection.

Several recent works \cite{pang2025poisafl,cox2024asynchronous,feng2021bafl,miao2023robust} also address robustness under delayed updates. However, approaches such as Catalyst~\cite{cox2024asynchronous} and PoiSAFL~\cite{pang2025poisafl} operate in semi-asynchronous settings, where the server must still wait for a batch of delayed updates. Other methods~\cite{feng2021bafl,miao2023robust} rely on blockchain or homomorphic encryption, which introduce substantial computational and deployment overhead. In contrast, \textbf{we focus on a fully asynchronous setting}, where the server updates the global model immediately upon receiving a single update, without relying on trusted data or heavy cryptographic assumptions.
Note that in~\cite{karimireddy2021learning}, history is incorporated via per-client momentum to improve robust aggregation, whereas our approach uses historical patterns to estimate missing or delayed updates, leading to a fundamentally different mechanism and objective.


\section{Threat Model}

\myparatight{Attacker’s goal and knowledge}%
Building on prior works~\cite{fang2020local,shejwalkar2021manipulating,cao2020fltrust}, we assume that the attacker controls a subset of malicious clients capable of either poisoning their local training data or directly altering their model updates to advance the attacker's objectives. The attacker may possess either full or partial knowledge of the FL system. In a full-knowledge attack, the attacker has access to all clients' model updates and is aware of the server's defense mechanism. In contrast, under a partial-knowledge attack, the attacker is limited to the model updates of malicious clients while still knowing the server's defense strategy. As noted in~\cite{fang2020local,shejwalkar2021manipulating}, full-knowledge attacks are significantly more potent than their partial-knowledge counterparts. Therefore, we employ the full-knowledge attack to assess the resilience of our proposed defense.

\myparatight{Defender’s goal and knowledge}%
The defender lacks any prior knowledge of the attacker's strategy or the number of malicious clients in the system. Our objective is to develop a reliable defense mechanism for asynchronous FL that meets two key criteria. First, in a benign environment where all clients behave honestly, the global model trained with our defense should achieve performance comparable to that of AsyncSGD~\cite{zheng2017asynchronous}, the state-of-the-art approach in such settings. Second, when facing poisoning attacks, the defense must effectively limit the impact of malicious clients, preserving the integrity of the learned model.

\section{The \alg Algorithm} 
\label{sec:alg}

\subsection{Overview}

Our proposed method, \alg, strengthens asynchronous FL systems against adversarial threats by systematically identifying malicious updates, reconstructing missing client contributions, and applying a resilient aggregation strategy. It begins by evaluating the consistency of local updates through their smoothness properties, discarding those that exhibit abrupt deviations from historical trends. To address the challenge of incomplete client participation, \alg approximates the missing updates by leveraging past model evolution, ensuring that the aggregation process remains balanced. Finally, the server combines both received and estimated updates using a Byzantine-robust aggregation mechanism that limits the influence of malicious manipulations. By seamlessly detecting anomalous updates, reconstructing missing client contributions, and employing a resilient aggregation strategy, \alg safeguards the learning process from malicious disruptions. This cohesive approach ensures that unreliable updates are excluded, plausible estimates compensate for incomplete participation, and the final aggregation remains robust, ultimately maintaining the integrity of the global model.

\subsection{Local Model Updates Filtering}

Our \alg incorporates a filtering mechanism grounded in the Lipschitz-smooth property of local model updates, leveraging their inherent smoothness to differentiate between benign and malicious contributions.
This property ensures that updates do not change too abruptly, thereby contributing to the stability of the learning process. To illustrate this in the context of our approach, consider a training round \( t \), where the server receives the local model update \( \bm{g}_i^{t-\tau_i} \) from client \( i \). Let \( \bm{g}_i^{ \varphi} \) denote the last local model update that client \( i \) sent to the server during training round \( \varphi \), where clearly, \( \varphi < t-\tau_i \). These earlier updates serve as a reference to track the evolution of client-side models over time. Additionally, let \( \bm{w}^{t-\tau_i} \) and \( \bm{w}^{\varphi} \) represent the corresponding global models used by client \( i \) to compute \( \bm{g}_i^{t-\tau_i} \) and \( \bm{g}_i^{\varphi} \), respectively.

To assess the consistency of the received update, the server calculates a Lipschitz factor for client \( i \) at round \( t \), denoted as \( \lambda_i^t \), which is defined as:
\begin{align}
\label{lipschitz_filter}
\lambda_i^t = \frac{\left\| \bm{g}_i^{t-\tau_i}  - \bm{g}_i^{ \varphi} \right\|}{\left\| \bm{w}^{t-\tau_i} - \bm{w}^{\varphi} \right\|}.
\end{align}

By computing \( \lambda_i^t \), the server quantifies the smoothness of the local model update, allowing it to evaluate whether the update aligns with expected behavior. To systematically track this metric, the server maintains a historical record of all computed Lipschitz factors. Let \( Q^t \) denote the list of Lipschitz factors up to round \( t \), capturing the evolution of local model updates over time. 
The motivation behind this approach is to identify anomalous or potentially malicious updates by comparing the smoothness of the current update against the historical distribution of client behaviors. Specifically, a local model update \( \bm{g}_i^{t-\tau_i} \) is deemed benign if it satisfies:
\begin{align}
\label{filter_con}
\lambda_i^t \le Q^t_{\alpha},
\end{align}
where \( Q^t_{\alpha} \) represents the \( \alpha \)-th percentile of values in \( Q^t \). This percentile-based threshold helps filter out abnormal updates, ensuring that only those within an expected range are accepted. By enforcing this constraint, our method effectively filters out abrupt or suspicious changes in model updates, which may indicate malicious behavior, such as data poisoning or model manipulation.

\begin{algorithm}[t] 
	\caption{Estimate the client's update.}
        \label{alg:estimate_update}
	\begin{algorithmic}[1]
		\renewcommand{\algorithmicrequire}{\textbf{Input:}}
		\renewcommand{\algorithmicensure}{\textbf{Output:}}
		\Require Client $k$; global models up to round $t$; model updates received up to round $t$; L-BFGS buffer size $\epsilon$. 
 
		\Ensure Estimated update $\hat{\bm{g}}_k^t$. 
        
        \State Update the L-BFGS buffers 
    $\bm{\Phi}^{t,\epsilon}$ and $\bm{\Pi}_k^{t,\epsilon}$.
    
        \State  $\Delta {\bm{w}}^{t} = \bm{w}^{t} - \bm{w}^{v}$.

        \State // We denote $\text{Diag}(\bm{Y})$ as the matrix consisting of the diagonal elements of $\bm{Y}$, and $\text{Tril}(\bm{Y})$ as the lower triangular portion of $\bm{Y}$, with $\bm{Y}^{\top}$ representing the transpose of $\bm{Y}$.
        
        \State  $\bm{Y}^{t,\epsilon}_k = (\bm{\Phi}^{t,\epsilon})^{\top} \bm{\Pi}_k^{t,\epsilon}$.
        
       \State $\bm{B}^{t,\epsilon}_k = \text{Diag}(\bm{Y}^{t,\epsilon}_k)$, $\bm{J}^{t,\epsilon}_k = \text{Tril}(\bm{Y}^{t,\epsilon}_k)$.
       
        \State  $\mu = ((\Delta \bm{g}^{t-1}_k)^{\top} \Delta \bm{w}^{t-1}) / ((\Delta \bm{w}^{t-1})^{\top}\Delta \bm{w}^{t-1})$.
        
        \State $\bm{l} = \begin{bmatrix}
        -\bm{B}_k^{t,\epsilon} & (\bm{J}_k^{t,\epsilon})^{\top} \\
        \bm{J}_k^{t,\epsilon} & \mu(\bm{\Phi}^{t,\epsilon})^{\top}\bm{\Phi}^{t,\epsilon}
        \end{bmatrix}^{-1}
        \begin{bmatrix}
        (\bm{\Pi}_k^{t,\epsilon})^{\top}\Delta\bm{w}^t \\
        \mu(\bm{\Phi}^{t,\epsilon})^{\top}\Delta\bm{w}^t 
        \end{bmatrix}$. 
        \State $\bm{H}_k^t\Delta\bm{w}^t = \mu \Delta\bm{w}^t - \begin{bmatrix}
        \bm{\Pi}_k^{t,\epsilon} &
        \mu\bm{\Phi}^{t,\epsilon}
        \end{bmatrix}\bm{l}$.
        \State $\hat{\bm{g}}_k^t =  \bm{g}_k^{v} + \bm{H}_k^t\Delta\bm{w}^t$.
	\end{algorithmic}
\end{algorithm}

\begin{algorithm}[t]
	\caption{Our \alg.}
	\label{our_alg}
	\begin{algorithmic}[1]
		\Statex  \underline{Server:} 
        
		\State Initializes the global model \(\bm{w}^{0}\) and distributes it to all clients.
        \State $Q \leftarrow \emptyset$.
		\For {$t=0,1,2,\cdots,T-1$} 
		\State After receiving the model update \(\bm{g}_i^{t-\tau_i}\) from client \(i\):
        \If {$t=0$}
        \State Apply $\ell_2$-norm clipping to $\bm{g}_i^0$ with threshold $G$: if $\left\|\bm{g}_i^0 \right\| > G$, rescale $\bm{g}_i^0$ such that $\left\|\bm{g}_i^0 \right\| = G$.
        \State Sets $\bm{g}^0 \leftarrow \bm{g}_i^0$.
         \Else
         \State Determines the Lipschitz factor \( \lambda_i^t \) from client \(i\) using Eq.~(\ref{lipschitz_filter}).
         \State $Q \leftarrow Q \bigcup \{\lambda_i^t\} $.
        \State Estimates the update \(\hat{\bm{g}}_k^t\) for each of the remaining \(n-1\) clients according to Eq.~(\ref{update_estimate}).
        \If {Eq.~(\ref{filter_con}) is satisfied}
         \State Computes the aggregated update $\bm{g}^t$ per Eq.~(\ref{agg_benign}).
        \Else
    	\State Computes the aggregated update $\bm{g}^t$ per Eq.~(\ref{agg_est}).

        \EndIf
        \EndIf
        \State Updates the global model as $\bm{w}^{t+1} =  \bm{w}^{t} -\eta \bm{g}^t$.
        \State Sends the updated global model $\bm{w}^{t+1} $ to client $i$.
		\EndFor
	\end{algorithmic} 
\end{algorithm}

\subsection{Local Model Updates Estimation}

Our \alg retains both the global model and the received local model updates from each training round. At round \( t \), upon receiving the local model update \( \bm{g}_i^{t-\tau_i} \) from client \( i \), the server does not immediately incorporate it into the global model. Instead, utilizing historical information such as stored global models and past updates, our \alg approximates the local model updates of the remaining \( n-1 \) clients.
Let \( \mathcal{S} \) represent the set of these clients, i.e., \( \mathcal{S} = [n] \setminus \{i\} \). For each client \( k \in \mathcal{S} \), \alg estimates its local model update at round \( t \), after which the server aggregates client \( i \)'s received update \( \bm{g}_i^{t-\tau_i} \) with the estimated updates from the other \( n-1 \) clients. The core challenge, therefore, is to accurately infer these missing updates based on historical information.

For a given client \( k \in \mathcal{S} \), let its most recent model update sent to the server be \( \bm{g}_k^{v} \), which was computed using a previous global model version \( \bm{w}^{v} \), where \( v < t \). Applying the Cauchy mean value theorem~\cite{lang2012real}, we estimate the local model update for client \( k \) at round \( t \) as:
\begin{align}
\label{update_estimate}
\hat{\bm{g}}_k^t =  \bm{g}_k^{v} + \bm{H}_k^t (\bm{w}^t - \bm{w}^{v}),
\end{align}
where \( \bm{H}_k^t \) represents the integrated Hessian matrix associated with client \( k \) at training round \( t \). It is derived by averaging the Hessian matrix along the trajectory between the past global model \( \bm{w}^v \) and the current global model \( \bm{w}^t \), computed as  
$
\bm{H}_k^t = \int_0^1 \bm{H}(\bm{w}^v + x(\bm{w}^t - \bm{w}^v)) \, dx
$.

In Eq.~(\ref{update_estimate}), it is evident that the server can estimate the updates from clients by leveraging both the stored historical data and the current global model at round $t$. However, directly calculating the estimated update $\hat{\bm{g}}_k^t$ from this equation is computationally demanding, primarily due to the need to compute the integrated Hessian matrix $\bm{H}_{k}^{t}$. To address this challenge, we propose an approximation of the integrated Hessian matrix using the well-known L-BFGS algorithm~\cite{byrd1995limited,byrd1994representations}. Rather than calculating the Hessian matrix explicitly, the L-BFGS method estimates it by utilizing a limited set of historical information from previous training rounds. Specifically, we define $\Delta {\bm{w}}^{t} = \bm{w}^{t} - \bm{w}^{v}$ as the global model difference at round $t$, and $\Delta {\bm{g}}_{k}^{t} = \hat{\bm{g}}_k^t - \bm{g}_k^v$ as the difference in the local model update for client $k$ at round $t$. Furthermore, $\bm{\Phi}^{t,\epsilon} = \{ \Delta {\bm{w}}^{t-\epsilon}, \Delta {\bm{w}}^{t-\epsilon+1}, \cdots, \Delta {\bm{w}}^{t-1} \}$ represents the set of global model differences from the past $\epsilon$ rounds, and $\bm{\Pi}_k^{t,\epsilon} = \{ \Delta {\bm{g}}_k^{t-\epsilon}, \Delta {\bm{g}}_k^{t-\epsilon+1}, \cdots, \Delta {\bm{g}}_k^{t-1} \}$ represents the local model update differences for client $k$ over the same period. The L-BFGS method uses these differences, along with the global model difference $\Delta {\bm{w}}^{t}$, to compute Hessian-vector products $\bm{H}_{k}^{t} \Delta {\bm{w}}^{t}$, which are then used to estimate the local model update $\hat{\bm{g}}_k^t$.
Algorithm~\ref{alg:estimate_update} presents the pseudocode for estimating the update of client $k$ during the training round $t$, where \( k \in \mathcal{S} \).

\subsection{Local Model Updates Aggregation}

After receiving client \( i \)'s update \( \bm{g}_i^{t-\tau_i} \), the server estimates the local updates of the remaining clients. To ensure robustness against poisoning attacks, it then applies a Byzantine-robust aggregation strategy, such as the coordinate-wise median~\cite{yin2018byzantine}, to aggregates these updates. Specifically, if Eq.~(\ref{filter_con}) holds, indicating that client \( i \)'s update is deemed reliable, the server integrates it with the estimated updates as:  
\begin{align}  
\label{agg_benign}  
\bm{g}^t = \text{Median}(\bm{g}_i^{t-\tau_i}, \{\hat{\bm{g}}_k^t\}_{k \in \mathcal{S}}),  
\end{align}  
where \( \bm{g}^t \) represents the aggregated update, \( \text{Median}(\cdot) \) denotes the coordinate-wise median aggregation~\cite{yin2018byzantine}, and \( \mathcal{S} \) refers to the set of the remaining \( n-1 \) clients, excluding client \( i \).  
Conversely, if Eq.~(\ref{filter_con}) is not satisfied, suggesting that client \( i \)'s update is likely malicious, the server disregards it and aggregates only the estimated updates:  
\begin{align}  
\label{agg_est}  
\bm{g}^t = \text{Median}(\{\hat{\bm{g}}_k^t\}_{k \in \mathcal{S}}).  
\end{align}  

It is important to note that instead of a simple averaging strategy, the server employs a robust aggregation mechanism in both Eq.~(\ref{agg_benign}) and Eq.~(\ref{agg_est}). This choice stems from the server's fundamental limitation: it lacks prior knowledge of the attacker's presence. Since malicious clients craft their updates strategically, their estimated updates remain adversarial, potentially degrading the global model. Therefore, using a robust aggregation rule helps mitigate the influence of these harmful updates, enhancing the security and reliability of the training process.
We also remark that if client \( i \)'s update is considered benign, it is not applied directly to the global model. Instead, it is combined with estimated updates from other clients (see Eq.~(\ref{agg_benign})). This approach ensures that valuable information embedded in benign clients' past update trajectories is preserved. By incorporating estimations based on historical updates, we enhance the robustness of \alg, as demonstrated in Table~\ref{variant_result}.

Algorithm~\ref{our_alg} presents the pseudocode of \alg, focusing solely on the server-side procedure. 
In the first communication round, no historical updates are available to compute the Lipschitz factor or construct the L-BFGS buffers. Therefore, the server applies $\ell_2$-norm clipping to the received updates to ensure bounded influence. Specifically, if the $\ell_2$ norm of a client’s update exceeds a predefined threshold $G$, the update is rescaled to have norm $G$.
In each of the following training rounds, upon receiving a local model update from client \(i\), the server first calculates the Lipschitz factor \( \lambda_i^t \) for the client. If the received update is deemed benign, the server incorporates it with the estimated updates; otherwise, only the estimated updates are aggregated. Finally, the server transmits the updated global model back to client \(i\).

\section{Theoretical Performance Analysis} \label{sec:analysis}

In this section, we provide a non-convex convergence guarantee for \alg under asynchronous updates, L-BFGS-based gradient estimation for non-uploading clients, and coordinate-wise median aggregation in the presence of Byzantine clients. Let $\mathcal{V}_t= \bigl\{\bm{g}_{i}^{t-\tau_i}\bigr\}\ \cup\ \bigl\{\hat{\bm{g}}_k^t: k\in\mathcal{S}\bigr\}$, which denotes the set of vectors input to the coordinate-wise median aggregation at round $t$, $\mathcal{H}$ be the set of benign clients of total clients.
We first introduce some necessary assumptions for analysis.

\begin{assumption}
\label{ass:smooth}
Each benign objective $f_j$ for \( j \in \mathcal{H} \) is differentiable and has $L$-Lipschitz gradient. That is, for all $\bm{w},\bm{v}\in\mathbb{R}^d$, and letting $F_{\mathcal H}(\bm w)=\frac{1}{m}\sum_{j\in\mathcal H} f_j(\bm w)$, we have:

\begin{align}
\|\nabla f_j(\bm{w})-\nabla f_j(\bm{v})\|\le L\|\bm{w}-\bm{v}\|. 
\end{align}
Consequently, the benign objective $F_\mathcal{H}$ is also $L$-smooth.
\end{assumption}

\begin{assumption}
\label{ass:lower}
$F_{\mathcal H}$ is bounded from below: 
\begin{align}
    \inf_{\bm{w}} F_{\mathcal H} = F_{\mathcal H}{\star} > -\infty.
\end{align}
\end{assumption}
\begin{assumption}
\label{ass:byz_fraction}
There are at most $b$ Byzantine clients and $m=n-b$ benign clients, and the median input size always strictly exceeds twice the Byzantine count:
\begin{align}
|\mathcal{V}_t| \ge 2b+1,\quad \forall t. 
\end{align}
In particular, since $|\mathcal{V}_t|\in\{n-1,n\}$ in \alg, a sufficient condition is
\begin{align}
b < \frac{n-1}{2}. 
\end{align}
\end{assumption}
\begin{assumption}
\label{ass:stale}
For any client $i \in [n]$, the delay $\tau_i$ associated with any local model update satisfies $\tau_i \le \tau_{\max}$, where $\tau_{\max}$ acts as the global upper bound on the delays across all clients in the system.
\end{assumption}
\begin{assumption}
\label{ass:benign_upload}
Whenever the client $i$ is benign, its update satisfies
\begin{align}
\bm{g}_{i}^{t-\tau_i} = \nabla f_{i}(\bm{w}^{t-\tau_i}) + \bm{\xi}_{i,t}, \label{eq:benign_upload_decomp}
\end{align}
where $\bm{\xi}_{i,t}$ is a zero-mean noise term conditioned on the past:
\begin{align}
\mathbb{E}[\bm{\xi}_{i,t} \mid \mathcal{F}_t] = \bm{0},\quad
\mathbb{E}[\|\bm{\xi}_{i,t}\|^2 \mid \mathcal{F}_t]\le \sigma^2. 
\end{align}
Here $\mathcal{F}_t$ denotes the sigma-algebra generated by the entire history up to $\bm{w}^t$.
\end{assumption}
\begin{assumption}
\label{ass:est}
For any benign client $j\in\mathcal{H}$, at any round $t$ in which $j\neq i$, the estimator output satisfies the second-moment bound
\begin{align}
\mathbb{E}\bigl[\|\hat{\bm{g}}_j^t - \nabla f_j(\bm{w}^t)\|^2 \mid \mathcal{F}_t\bigr] \le \varepsilon_{\mathrm{est}}^2. \label{eq:est_acc}
\end{align}
The constant $\varepsilon_{\mathrm{est}}$ depends on the L-BFGS buffer size, curvature variation, and how often the client updates are refreshed.
\end{assumption}

\begin{assumption}
\label{ass:filter_bounded}
Byzantine behavior is arbitrary in communication; however, due to the Lipschitz filter and/or an explicit server-side clipping rule, all vectors that enter the median satisfy a uniform second-moment bound:
\begin{align}
\mathbb{E}\bigl[\|\bm{v}\|^2 \mid \mathcal{F}_t\bigr] \le G^2,\quad \forall \bm{v}\in \mathcal{V}_t,\ \forall t. \label{eq:input_norm_bound}
\end{align}
Moreover, the number of Byzantine vectors in $\mathcal{V}_t$ is at most $b$.
\end{assumption}
\begin{assumption}
\label{ass:heterogeneity}
There exists $\zeta\ge 0$ such that for all $\bm{w}$,
\begin{align}
\frac{1}{m}\sum_{j\in\mathcal{H}} \|\nabla f_j(\bm{w}) - \nabla F_{\mathcal H}(\bm{w})\|^2 \le \zeta^2. \label{eq:heterogeneity}
\end{align}
\end{assumption}

\begin{thm}[Convergence of \alg under bounded tracking error]
\label{thm:main}
Let Assumptions~\ref{ass:smooth}--\ref{ass:heterogeneity} hold. Suppose the stepsize satisfies
\begin{align}
0<\eta \le \frac{1}{4L}. \label{eq:stepsize_cond}
\end{align}
Then for any $T\ge 1$,
\begin{align}
\frac{1}{T}\sum_{t=0}^{T-1}\mathbb{E}\bigl[\|\nabla F_{\mathcal H}(\bm{w}^t)\|^2\bigr]
\le
\frac{4\bigl(F_{\mathcal H}(\bm{w}^0)-F_{\mathcal{H}}{\star}\bigr)}{\eta T}
+
4E_{\mathrm{track}}^2, \label{eq:main_rate}
\end{align}
where $E_{\mathrm{track}}^2$ can be chosen as
\begin{align}
E_{\mathrm{track}}^2
\le
C_{\mathrm{med}}\Bigl(
\zeta^2
+\varepsilon_{\mathrm{est}}^2
+\sigma^2
+L^2\eta^2\tau_{\max}^2 G^2
\Bigr), 
\end{align}
for an absolute constant $C_{\mathrm{med}}>0$ that depends only on the coordinate-wise median bound used in the analysis.
\end{thm}
\begin{proof}
The proof is relegated to Appendix~\ref{sec:appendix_main}.
\end{proof}
\begin{corollary}[Diminishing stepsize]
\label{cor:diminishing}
Under the conditions of Theorem~\ref{thm:main}, choose $\eta = \min\{\frac{1}{4L},\, \frac{c}{\sqrt{T}}\}$ for any $c>0$. Then
\begin{align}
\frac{1}{T}\sum_{t=0}^{T-1}\mathbb{E}\|\nabla F_{\mathcal H}(\bm{w}^t)\|^2
=
\mathcal{O}\Bigl(\frac{1}{\sqrt{T}}\Bigr)
+
\mathcal{O}\bigl(E_{\mathrm{track}}^2\bigr). 
\end{align}
\end{corollary}
\begin{remark}
Assumption~\ref{ass:byz_fraction} is required because the server uses coordinate-wise median aggregation; a strict benign majority in the aggregated set $\mathcal{V}_t$ is sufficient to ensure median robustness in each coordinate. Since \alg may aggregate $n-1$ vectors (when the received update is rejected), a sufficient condition is $b<(n-1)/2$.
The $\ell_2$-norm clipping applied at $t = 0$ (Algorithm~\ref{our_alg}) serves only as initialization and is consistent with the boundedness condition in Assumption~\ref{ass:filter_bounded}; hence it does not affect the convergence analysis for $t \geq 1$.
\end{remark}

\begin{remark}
The term $E_{\mathrm{track}}^2$ quantifies the combined effect of (i) inter-client heterogeneity ($\zeta^2$), (ii) estimation error for non-uploading clients ($\varepsilon_{\mathrm{est}}^2$), (iii) stochastic noise from the single fresh upload ($\sigma^2$), and (iv) asynchrony through staleness ($L^2\eta^2\tau_{\max}^2 G^2$).
\end{remark}


\section{Experimental Evaluation}  \label{sec:exp}

\subsection{Experimental Setup}
\subsubsection{Datasets}

\begin{table}[t]
	\caption{CNN architecture.}
	\centering
	\begin{tabular}{|c|c|} \hline 
		{Layer} & {Size} \\ \hline
		{Input} & { $28\times28\times1$}\\ \hline
		{Convolution + ReLU} & { $3\times3\times30$}\\ \hline
		{Max Pooling} & { $2\times2$}\\ \hline
		{Convolution + ReLU} & { $3\times3\times50$}\\ \hline
		{Max Pooling} & { $2\times2$}\\ \hline
		{Fully Connected + ReLU} & {100}\\ \hline
		{Softmax} & {10}\\ \hline
	\end{tabular}
	\label{cnn_arch}
\end{table}

We conducted experiments on a diverse selection of datasets. These include Fashion-MNIST~\cite{xiao2017online}, CIFAR-10~\cite{cifar10data}, CIFAR-100~\cite{cifar10data}, and Tiny-ImageNet~\cite{deng2009imagenet} for image classification, along with the Udacity dataset~\cite{Udacity}, which contains real-world data from autonomous driving environments.
Comprehensive information about these datasets can be found in Appendix~\ref{app_dataset}.

\subsubsection{Poisoning Attacks}

Our evaluation incorporates five untargeted attacks, including label flipping attack~\cite{tolpegin2020data}, SignFlip attack~\cite{fang2020local}, Gaussian attack~\cite{blanchard2017machine}, Min-Max attack~\cite{shejwalkar2021manipulating}, and Adaptive attack~\cite{shejwalkar2021manipulating}, as well as five targeted attacks (e.g., backdoor attacks), namely Scaling attack~\cite{bagdasaryan2020backdoor}, DBA attack~\cite{xie2019dba}, Projected gradient descent attack~\cite{sun2019can}, Neurotoxin attack~\cite{zhang2022neurotoxin}, and 3DFed attack~\cite{li20233dfed}. 
Detailed descriptions of these poisoning attacks are provided in Appendix~\ref{app_attack}.

\begin{table*}[h]
\caption{Performance of different methods on the Fashion-MNIST, CIFAR-10, CIFAR-100, and Tiny-ImageNet datasets. For untargeted attacks, results are reported as TER, while for targeted attacks (e.g., backdoor attacks), results are reported as TER/ASR. Lower values indicate better defense performance. Results on the Udacity dataset are reported in Table~\ref{main_result_Udacity}.}
\centering
\addtolength{\tabcolsep}{-1.58pt}
  \subfloat[Fashion-MNIST dataset.]
  {
    \begin{tabular}{|c|c|c|c|c|c|c|c|c|c|c|c|}
    \hline
       Method    & \text{No attack} & \text{Labelflip} & \text{Signflip} & \text{Gaussian} & \text{Scaling} & \text{DBA} & \text{PGD} & \text{Neurotoxin} & \text{3DFed} & \text{Min-Max} & \text{Adaptive} \\ \hline
            \text{AsyncSGD} & 0.12 & 0.12 & 0.41 & 0.52 & 0.41/0.69 & 0.53/0.60 & 0.28/0.36 & 0.25/0.25 & 0.90/1.00 & 0.90 & 0.90 \\ \hline
            \text{Kardam} & 0.17 & 0.17 & 0.18 & 0.15 & 0.79/0.92 & 0.45/0.54 & 0.36/0.36 & 0.35/0.36 & 0.35/0.36 & 0.62 & 0.90 \\ \hline
            \text{BASGD} & 0.14 & 0.16 & 0.18 & 0.73 & 0.36/0.77 & 0.33/0.65 & 0.46/0.53 & 0.34/0.39 & 0.35/1.00 & 0.79 & 0.90 \\ \hline
            \text{Sageflow} & 0.14 & 0.14 & 0.16 & 0.90 & 0.70/0.85 & 0.60/0.69 & 0.31/0.41 & 0.32/0.35 & 0.90/0.88 & 0.90 & 0.90 \\ \hline
            \text{Zeno++} & 0.17 & 0.18 & 0.19 & 0.17 & 0.21/0.24 & 0.21/0.21 & 0.21/0.21 & 0.22/0.21 & 0.21/0.20 & 0.30 & 0.33 \\ \hline
            \text{AFLGuard} & 0.22 & 0.27 & 0.34 & 0.26 & 0.28/0.36 & 0.26/0.36 & 0.26/0.15 & 0.26/0.20 & 0.26/0.12 & 0.33 & 0.35 \\ \hline
            \rowcolor{greyL}
            \alg & 0.12 & 0.12 & 0.13 & 0.14 & 0.18/0.07 & 0.17/0.05 & 0.16/0.07 & 0.15/0.06 & 0.17/0.03 & 0.15 & 0.18 \\ \hline
    \end{tabular}
    }
    \vspace{-0.08in}
    \\
  \subfloat[CIFAR-10 dataset.]
  {
        \begin{tabular}{|c|c|c|c|c|c|c|c|c|c|c|c|}
        \hline
              Method     & \text{No attack} & \text{Labelflip} & \text{Signflip} & \text{Gaussian} & \text{Scaling} & \text{DBA} & \text{PGD} & \text{Neurotoxin} & \text{3DFed} & \text{Min-Max} & \text{Adaptive} \\ \hline
        \text{AsyncSGD} & 0.22 & 0.37 & 0.48 & 0.80 & 0.85/1.00 & 0.86/1.00 & 0.70/0.97 & 0.56/0.88 & 0.90/1.00 & 0.84 & 0.86 \\ \hline
        \text{Kardam} & 0.25 & 0.33 & 0.33 & 0.45 & 0.56/0.42 & 0.90/1.00 & 0.49/0.68 & 0.71/0.71 & 0.48/0.29 & 0.82 & 0.87 \\ \hline
        \text{BASGD} & 0.29 & 0.39 & 0.42 & 0.81 & 0.87/0.97 & 0.85/0.97 & 0.73/0.43 & 0.70/0.56 & 0.90/1.00 & 0.88 & 0.89 \\ \hline
        \text{Sageflow} & 0.28 & 0.48 & 0.58 & 0.83 & 0.89/1.00 & 0.87/1.00 & 0.85/0.99 & 0.79/0.92 & 0.90/1.00 & 0.85 & 0.81 \\ \hline   
        \text{Zeno++} & 0.23 & 0.32 & 0.33 & 0.33 & 0.63/0.54 & 0.56/0.51 & 0.52/0.49 & 0.43/0.37 & 0.53/0.39 & 0.34 & 0.32 \\ \hline
        \text{AFLGuard} & 0.25 & 0.36 & 0.37 & 0.37 & 0.44/0.26 & 0.42/0.30 & 0.48/0.32 & 0.45/0.32 & 0.46/0.17 & 0.39 & 0.37 \\ \hline
        \rowcolor{greyL}
        \alg & 0.23 & 0.29 & 0.29 & 0.30 & 0.29/0.03 & 0.26/0.05 & 0.30/0.10 & 0.26/0.07 & 0.28/0.01 & 0.32 & 0.35 \\ \hline
        \end{tabular}
  }
  \\
    \vspace{-0.08in}
    \subfloat[CIFAR-100 dataset.]
  {
    \begin{tabular}{|c|c|c|c|c|c|c|c|c|c|c|c|}
    \hline
       Method    & \text{No attack} & \text{Labelflip} & \text{Signflip} & \text{Gaussian} & \text{Scaling} & \text{DBA} & \text{PGD} & \text{Neurotoxin} & \text{3DFed} & \text{Min-Max} & \text{Adaptive} \\ \hline
        \text{AsyncSGD} & 0.43 & 0.53 & 0.66 & 0.90 & 0.95/1.00 & 0.95/1.00 & 0.67/0.52 & 0.61/0.61 & 0.95/1.00 & 0.95 & 0.94 \\ \hline
        \text{Kardam} & 0.47 & 0.48 & 0.54 & 0.61 & 0.93/1.00 & 0.93/1.00 & 0.71/0.54 & 0.68/0.38 & 0.63/0.37 & 0.95 & 0.91 \\ \hline
        \text{BASGD} & 0.52 & 0.53 & 0.64 & 0.90 & 0.89/0.95 & 0.91/0.94 & 0.79/0.26 & 0.76/0.15 & 0.95/1.00 & 0.95 & 0.95 \\ \hline
        \text{Sageflow} & 0.56 & 0.59 & 0.75 & 0.91 & 0.87/0.97 & 0.89/0.99 & 0.68/0.13 & 0.61/0.13 & 0.95/1.00 & 0.94 & 0.93 \\ \hline   
        \text{Zeno++} & 0.48 & 0.51 & 0.51 & 0.60 & 0.21/0.68 & 0.63/0.15 & 0.72/0.19 & 0.74/0.26 & 0.67/0.17 & 0.63 & 0.83 \\ \hline
        \text{AFLGuard} & 0.43 & 0.54 & 0.50 & 0.47 & 0.65/0.18 & 0.67/0.48 & 0.67/0.74 & 0.70/0.38 & 0.68/0.53 & 0.92 & 0.85 \\ \hline
        \rowcolor{greyL}
        \alg & 0.43 & 0.46 & 0.49 & 0.44 & 0.53/0.03 & 0.52/0.07 & 0.49/0.02 & 0.53/0.09 & 0.51/0.05 & 0.53 & 0.56 \\ \hline
        \end{tabular}
  }
    \vspace{-0.08in}
    \subfloat[Tiny-ImageNet dataset.]
    {
    \begin{tabular}{|c|c|c|c|c|c|c|c|c|c|c|c|}
    \hline
           Method    & \text{No attack} & \text{Labelflip} & \text{Signflip} & \text{Gaussian} & \text{Scaling} & \text{DBA} & \text{PGD} & \text{Neurotoxin} & \text{3DFed} & \text{Min-Max} & \text{Adaptive} \\ \hline
    \text{AsyncSGD} & 0.61 & 0.68 & 0.75 & 0.98 & 0.62/1.00 & 0.62/1.00 & 0.62/0.44 & 0.62/0.77 & 0.62/1.00 & 0.98 & 0.60 \\ \hline
    \text{Kardam} & 0.63 & 0.66 & 0.72 & 0.76 & 0.76/0.03 & 0.76/0.05 & 0.75/0.05 & 0.76/0.03 & 0.76/0.03 & 0.98 & 0.63 \\ \hline
    \text{BASGD} & 0.63 & 0.64 & 0.73 & 0.64 & 0.56/0.39 & 0.56/0.36 & 0.56/0.08 & 0.56/0.10 & 0.56/0.02 & 0.98 & 0.63 \\ \hline
    \text{Sageflow} & 0.63 & 0.68 & 0.74 & 0.99 & 0.98/0.97 & 0.98/0.98 & 0.83/0.43 & 0.83/0.36 & 0.99/1.00 & 0.99 & 0.99 \\ \hline    
    \text{Zeno++} & 0.62 & 0.65 & 0.65 & 0.66 & 0.64/0.03 & 0.65/0.05 & 0.63/0.03 & 0.65/0.05 & 0.63/0.03 & 0.67 & 0.64 \\ \hline
    \text{AFLGuard} & 0.60 & 0.68 & 0.72 & 0.63 & 0.61/0.03 & 0.62/0.03 & 0.62/0.02 & 0.60/0.05 & 0.60/0.03 & 0.62 & 0.65 \\ \hline
    \rowcolor{greyL}
    \alg & 0.58 & 0.60 & 0.66 & 0.62 & 0.61/0.03 & 0.61/0.03 & 0.62/0.03 & 0.62/0.03 & 0.64/0.03 & 0.60 & 0.62 \\ \hline
    \end{tabular}
  }
  \label{main_result}
  \vspace{-.1in}
\end{table*}

\begin{table}[h]
\centering
\addtolength{\tabcolsep}{-2.5pt}
\caption{RMSE of different methods on Udacity dataset.}
\begin{tabular}{|c|c|c|c|c|c|}
\hline
     Method      & No attack & Signflip & Gaussian & Min-Max & Adaptive \\ \hline
AsyncSGD   & 0.17      & 0.29     & 1.10      & 0.36    & 0.33             \\ \hline
Kardam     & 0.18      & 0.28     & 0.19     & 0.55    & 0.24             \\ \hline
BASGD      & 0.17      & 0.19     & 0.19     & inf     & 0.20             \\ \hline
Zeno++     & 0.17      & inf      & 0.24     & inf     & inf              \\ \hline
AFLGuard   & 0.18      & 0.43     & 0.25     & 0.19    & 0.26             \\ \hline
\rowcolor{greyL}
\alg & 0.17    & 0.17     & 0.18     & 0.17    & 0.17             \\ \hline
\end{tabular}
  \label{main_result_Udacity}
  \vspace{-.1in}
\end{table}

\subsubsection{Compared Methods}
By default, we compare \alg with six baselines: AsyncSGD~\cite{zheng2017asynchronous}, Kardam~\cite{damaskinos2018asynchronous}, BASGD~\cite{yang2021basgd}, Sageflow~\cite{park2021sageflow}, Zeno++~\cite{xie2020zeno++}, and AFLGuard~\cite{fang2022aflguard}. 
Details of these six methods are provided in Appendix~\ref{app_method}.
Note that we also include a comparison between \alg and the more recent defense AsyncDefender~\cite{bai2025asyncdefender}, as reported in Table~\ref{tab:asyncdefender}.

\subsubsection{Evaluation Metrics}

We evaluate defense effectiveness using task-specific metrics: testing error rate (TER) and attack success rate (ASR) for image classification (Fashion-MNIST, CIFAR-10, CIFAR-100, and Tiny-ImageNet), and root mean squared error (RMSE) for regression (Udacity). TER measures the proportion of clean test samples that are misclassified, while ASR quantifies the fraction of targeted examples incorrectly classified into the target label. For regression, RMSE is defined as $\text{RMSE} = \sqrt{\frac{1}{M} \sum_{i=1}^{M} (\bar{y}_i - y_i)^2}$, where $\bar{y}_i$ and $y_i$ denote the predicted and true values, respectively, and $M$ is the number of test instances. We exclude targeted attacks on Udacity, as existing attack methods are not designed for regression tasks. For TER, ASR, and RMSE, lower values indicate better defense performance.

\subsubsection{Non-IID Setting, and Parameter Settings} \label{setups}
A fundamental characteristic of FL is the non-independent and non-identically distributed (Non-IID) nature of client training data. Following~\cite{fang2020local}, we simulate Non-IID distributions for the Fashion-MNIST, CIFAR-10, CIFAR-100, and Tiny-ImageNet datasets as follows: Given a dataset with \( z \) classes, clients are randomly divided into \( z \) groups. A training sample with label \( q \) is assigned to clients in group \( q \) with probability \( x \), while clients in other groups receive it with probability \( (1 - x) / (z - 1) \). A higher \( x \) value increases the Non-IID nature of the data distribution, and we set \( x = 0.5 \) by default. Notably, the Udacity dataset inherently exhibits Non-IID characteristics, eliminating the need for additional simulation.

We set up 50 clients for Fashion-MNIST, CIFAR-10, and CIFAR-100 datasets, and 10 clients for Tiny-ImageNet and Udacity. In the Udacity dataset, each client represents an autonomous driving company. By default, 20\% of clients are malicious. For backdoor attacks, we insert a $4 \times 4$ square trigger with random pixel values in the bottom-right corner of each image across all datasets. The batch size is 32 for Fashion-MNIST and 64 for the others. We train for 20,000 rounds on Fashion-MNIST and Udacity, 30,000 on CIFAR-10 and CIFAR-100, and 80,000 on Tiny-ImageNet, using a learning rate of 0.01. A CNN is used for Fashion-MNIST (see Table~\ref{cnn_arch} for the CNN architecture), and ResNet-20~\cite{he2016deep} for the rest. 
To simulate asynchronous FL, following ~\cite{fang2022aflguard,xie2020zeno++}, the server maintains all previous global models. In each round, the server randomly selects a client and uniformly samples a global model from the past $[0, \tau_{\max}]$ rounds, where $\tau_{\max}$ is the maximum delay and is set to 10 by default. The sampled global model is then sent to the selected client for local training.
For all attack settings, attacks start from the first round and persist in every subsequent round. For methods such as~\cite{park2021sageflow,xie2020zeno++,fang2022aflguard} that require a server-side trusted dataset, we uniformly sample 100 clean examples from the union of all clients’ local training data, following prior work~\cite{fang2022aflguard}.
In \alg, we set $\alpha = 0.8$ and buffer size $\epsilon = 3$.
Following~\cite{karimireddy2021learning}, we use grid search to determine the value of $G$. Specifically, we search over $G \in \{20, 50, 100\}$. We observe that the performance is stable across different choices of $G$, indicating that the method is not sensitive to this parameter. Based on this observation, we fix $G = 50$ in all experiments.
Experiments were conducted using NVIDIA V100 GPUs, with each test executed five times. 
We report the average results, 
and the standard deviations are all within 0.03.

\subsection{Experimental Results}

\begin{figure*}[t]
	\centering
	\includegraphics[scale = 0.35]{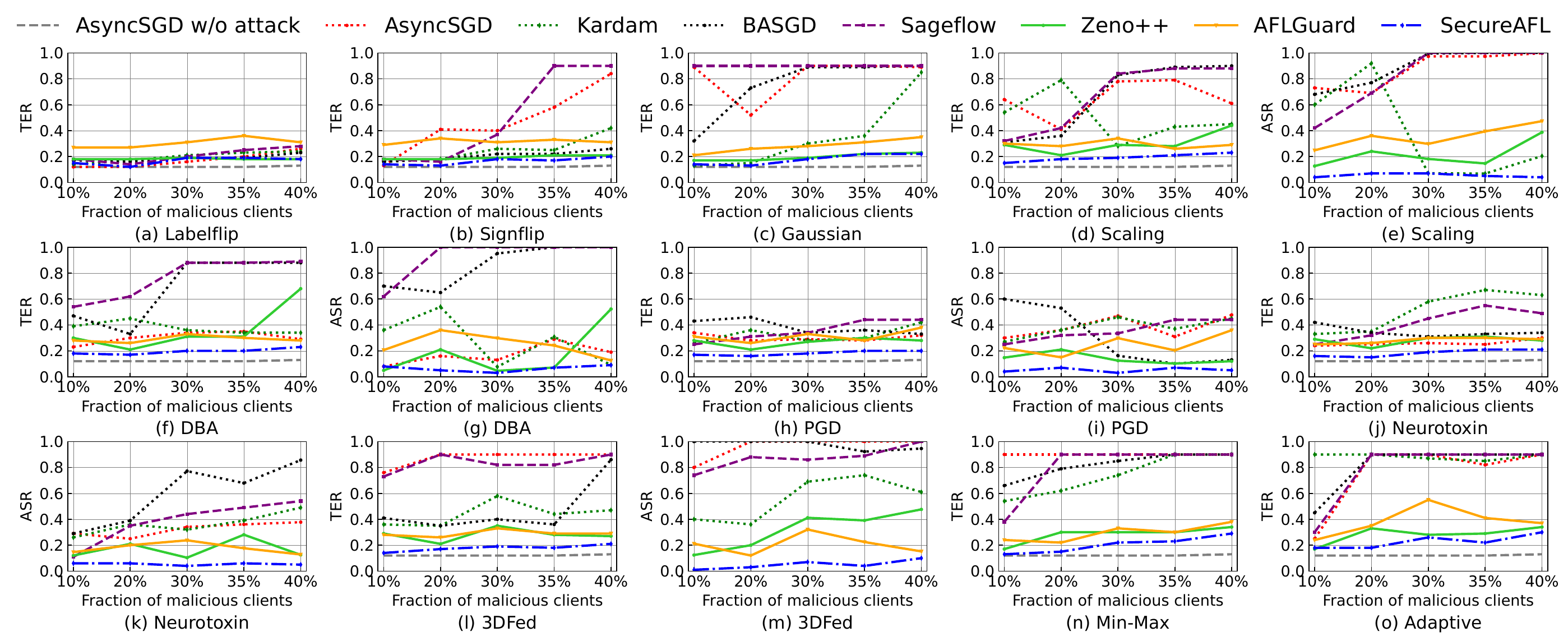}
	\caption{Impact of fraction of malicious clients on Fashion-MNIST dataset.}
	\label{malicious_fraction}
\end{figure*}

\myparatight{\alg is Effective}%
Table~\ref{main_result} summarizes the performance of all compared methods on four image classification benchmarks, while Table~\ref{main_result_Udacity} reports results on the Udacity regression task. Across all datasets, \alg consistently matches or closely tracks the performance of AsyncSGD in benign settings, demonstrating that the proposed defense does not introduce unnecessary bias or degradation when no adversarial behavior is present. For example, under the “No attack” condition, \alg achieves testing error rates of 0.12 on Fashion-MNIST, 0.23 on CIFAR-10, 0.43 on CIFAR-100, and 0.58 on Tiny-ImageNet, all of which are comparable to or slightly better than those of AsyncSGD and other baselines. This confirms that the filtering, estimation, and robust aggregation components of \alg do not hinder convergence or accuracy in non-adversarial environments. On the Udacity dataset, \alg similarly attains an RMSE of 0.17, matching the best-performing baselines and indicating that the framework generalizes well beyond classification tasks to regression scenarios 

In adversarial settings, \alg exhibits clear and consistent advantages over existing asynchronous FL defenses. Under untargeted attacks such as Labelflip, Signflip, and Gaussian noise, \alg maintains low and stable error rates across all datasets, whereas AsyncSGD and several baselines experience severe performance degradation. For instance, on CIFAR-10 under the Gaussian attack, \alg limits the testing error to 0.30, compared to 0.80 for AsyncSGD and over 0.80 for BASGD. Similar trends are observed on CIFAR-100 and Tiny-ImageNet, where \alg substantially suppresses error growth even when attacks are strong. For targeted backdoor attacks, \alg is particularly effective in simultaneously controlling both testing error rate (TER) and attack success rate (ASR). Across Scaling, DBA, PGD, Neurotoxin, and 3DFed attacks, \alg consistently reduces ASR to near-zero levels (often below 0.10), while maintaining TER close to benign performance. This dual robustness is notably absent in other defenses, many of which either fail to suppress ASR or do so at the cost of significantly inflated TER. The effectiveness of \alg is further corroborated by Figure~\ref{malicious_fraction}, which shows that even as the fraction of malicious clients increases to 40\%, \alg sustains the lowest test error and attack success rates among all methods. These results collectively demonstrate that \alg provides strong and reliable protection against a wide spectrum of poisoning attacks, including adaptive adversaries, while preserving high model utility across diverse datasets and system conditions.

\begin{figure*}[!t]
	\centering
	\includegraphics[scale = 0.35]{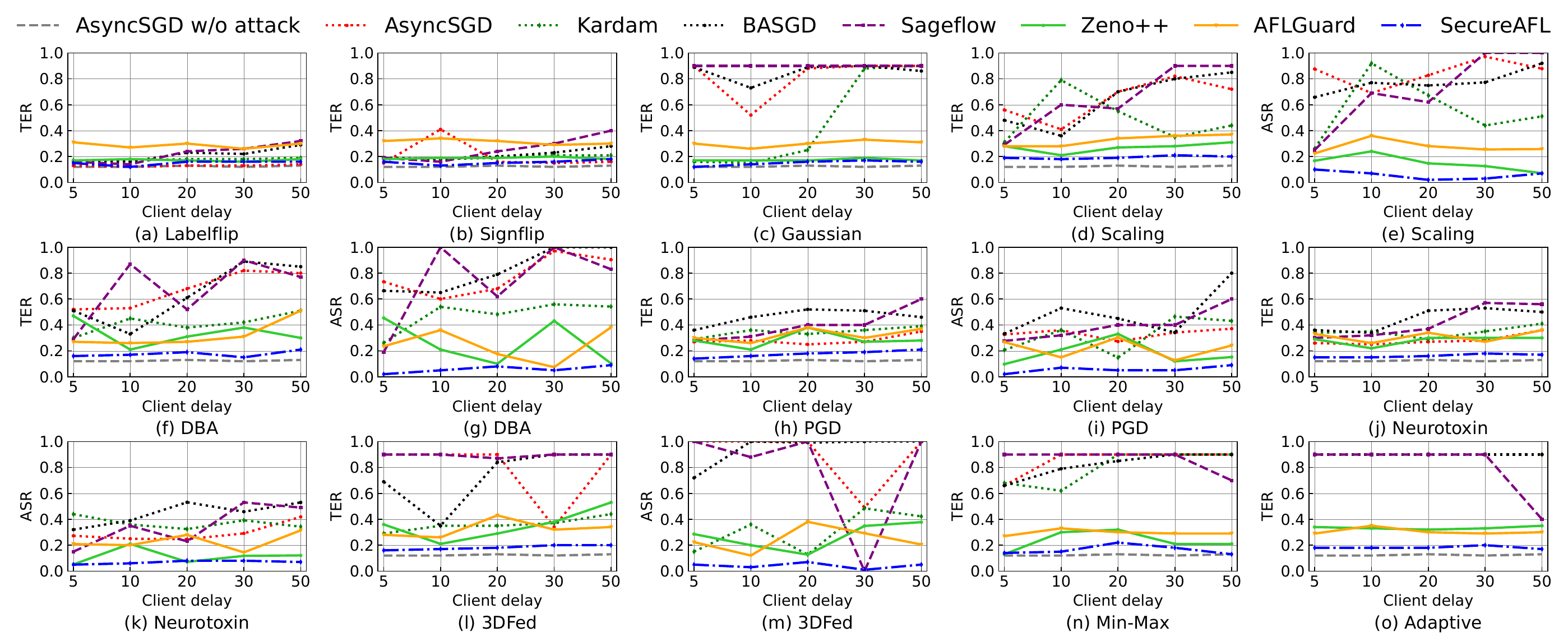}
	\caption{Impact of client delay on Fashion-MNIST dataset.}
	\label{delay_impact}
\end{figure*}

\myparatight{Impact of the fraction of malicious clients}%
Fig.~\ref{malicious_fraction} evaluates the robustness of \alg as the proportion of malicious clients increases from 10\% to 40\% on the Fashion-MNIST dataset under a wide range of poisoning attacks. As the fraction of malicious participants grows, most baseline methods exhibit a rapid deterioration in performance, reflected by sharply increasing testing error rates and, for targeted attacks, near-saturated attack success rates. In contrast, \alg demonstrates a markedly slower degradation trend and consistently maintains superior performance across all evaluated attack types. Even at moderate adversarial levels (20\%–30\% malicious clients), \alg preserves testing error rates close to those observed in non-adversarial AsyncSGD, while competing defenses such as Kardam and BASGD already show substantial instability. When the malicious fraction reaches 40\%, a particularly challenging regime for Byzantine-robust FL, \alg continues to outperform all baselines, achieving the lowest testing error among the compared methods and the lowest attack success rates for backdoor-based attacks. Notably, while AsyncSGD and Sageflow experience near-complete compromise under strong attacks such as Scaling, DBA, and PGD, \alg effectively suppresses malicious influence, preventing both widespread accuracy collapse and targeted misclassification. These results indicate that \alg scales robustly with adversarial strength and can tolerate a high proportion of malicious clients, highlighting the effectiveness of its combined filtering, update estimation, and robust aggregation mechanisms in heavily adversarial asynchronous FL environments.

\myparatight{Impact of the client delay}%
Fig.~\ref{delay_impact} investigates the robustness of different asynchronous FL defenses as the maximum client delay increases from 5 to 50 on the Fashion-MNIST dataset under various poisoning attacks. As client delays grow, the adverse effects of update staleness become increasingly pronounced, significantly challenging the stability of many baseline methods. In particular, Kardam and BASGD exhibit strong sensitivity to delayed updates, with their testing error rates escalating rapidly as delays exceed moderate levels. For example, under the Gaussian attack, Kardam’s testing error increases sharply when the maximum delay rises beyond 20, eventually approaching near-random performance. Similar instability trends are observed for BASGD across multiple attack scenarios, indicating that buffering-based or deviation-threshold defenses struggle to reliably distinguish benign but stale updates from malicious ones in highly asynchronous environments.

In contrast, \alg maintains consistently stable performance across the entire range of client delays. Even when the maximum delay reaches 50, \alg preserves low testing error rates and, for targeted attacks, low attack success rates, demonstrating strong resilience to severe asynchrony. This robustness stems from \alg’s explicit modeling of delayed and missing client updates through historical estimation, which mitigates the impact of stale information, as well as its Lipschitz-based filtering mechanism that remains effective regardless of delay magnitude. Compared to methods such as Sageflow, Zeno++, and AFLGuard, whose performance degrades notably as delays increase, often due to their reliance on trusted data or sensitivity to stale gradients, \alg consistently delivers superior stability. Overall, these results highlight \alg’s ability to handle extreme client delays, making it particularly well-suited for real-world FL deployments characterized by heterogeneous computation and communication latencies.

\begin{figure*}[!t]
	\centering
	\includegraphics[scale = 0.35]{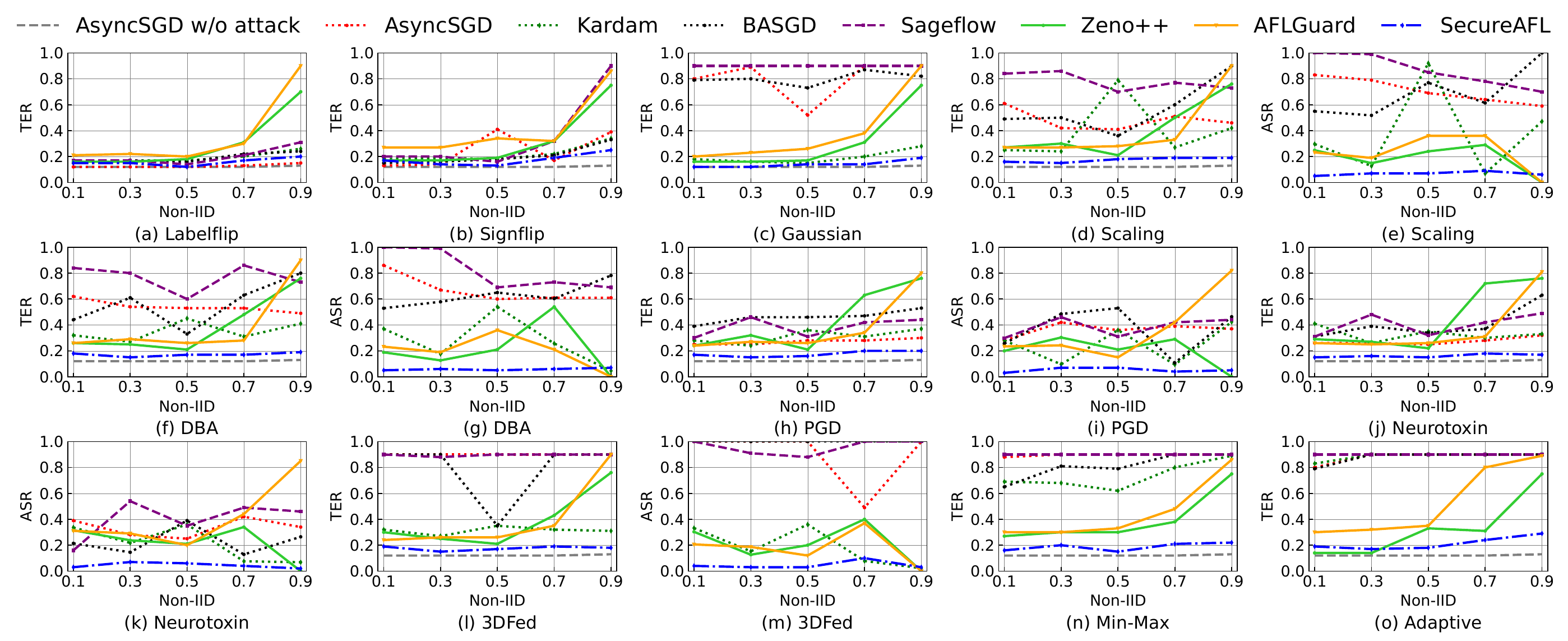}
	\caption{Impact of degree of Non-IID on Fashion-MNIST dataset.}
	\label{delay_non_iid}
\end{figure*}

\myparatight{Impact of total number of clients}%
Fig.~\ref{number_clients} in the Appendix examines the scalability of \alg by varying the total number of participating clients from 50 to 300, while fixing the fraction of malicious clients at 20\%. As the system scale increases, the learning dynamics of asynchronous FL become more complex due to higher heterogeneity, increased asynchrony, and a larger volume of potentially malicious updates. Many baseline methods struggle under these conditions, exhibiting noticeable performance fluctuations as the client population grows. In particular, defenses such as Kardam and BASGD show unstable behavior, with testing error rates increasing or oscillating as the number of clients rises, indicating limited scalability in large-scale federated settings.

In contrast, \alg demonstrates consistently stable and robust performance across all evaluated client counts. Its testing error remains largely unchanged as the number of clients increases, indicating that the proposed framework effectively scales with system size. This stability can be attributed to \alg’s design, which aggregates both received and estimated updates using a Byzantine-robust rule, ensuring that the influence of malicious clients does not grow disproportionately with the number of participants. Moreover, the update estimation mechanism allows \alg to maintain a balanced aggregation process even when only a subset of clients actively contributes at each round, a scenario that becomes increasingly common in large-scale asynchronous systems. Overall, these results confirm that \alg scales effectively to larger FL deployments and remains resilient to poisoning attacks even as the number of participating clients grows substantially.

\begin{table*}[h]
\caption{Impact of $\alpha$ on Fashion-MNIST dataset.}
\centering
\begin{tabular}{|c|c|c|c|c|c|c|c|c|c|c|c|}
\hline
{$\alpha$} & {No attack} & {Labelflip} & {Signflip} & {Gaussian} & {Scaling} & {DBA} & {PGD} & {Neurotoxin} & {3DFed} & {Min-Max} & {Adaptive} \\ \hline

0.4 & 0.13 & 0.15 & 0.14 & 0.14 & 0.19/0.02 & 0.16/0.02 & 0.17/0.07 & 0.19/0.08 & 0.19/0.10 & 0.22 & 0.23 \\ \hline
0.6 & 0.15 & 0.16 & 0.15 & 0.14 & 0.18/0.08 & 0.17/0.03 & 0.19/0.10 & 0.20/0.01 & 0.17/0.06 & 0.20 & 0.23 \\ \hline
0.8 & 0.12 & 0.12 & 0.13 & 0.14 & 0.18/0.07 & 0.17/0.05 & 0.16/0.07 & 0.15/0.06 & 0.17/0.03 & 0.15 & 0.18 \\ \hline
\end{tabular}
\label{alpha_impact}
\end{table*}

\begin{table*}[h]
\caption{Different variants of \alg on Fashion-MNIST dataset.}
\centering
\addtolength{\tabcolsep}{-1.2pt}
\begin{tabular}{|c|c|c|c|c|c|c|c|c|c|c|c|}
\hline
{Variant} & {No attack} & {Labelflip} & {Signflip} & {Gaussian} & {Scaling} & {DBA} & {PGD} & {Neurotoxin} & {3DFed} & {Min-Max} & {Adaptive} \\ \hline
\text{Variant I} & 0.12 & 0.14 & 0.15 & 0.14 & 0.40/0.36 & 0.34/0.33 & 0.40/0.19 & 0.30/0.22 & 0.38/0.34 & 0.35 & 0.35 \\ \hline
\text{Variant II} & 0.17 & 0.17 & 0.20 & 0.18 & 0.72/0.90 & 0.44/0.58 & 0.34/0.38 & 0.37/0.34 & 0.35/0.40 & 0.87 & 0.90 \\ \hline
\text{Variant III} & 0.16 & 0.16 & 0.19 & 0.20 & 0.34/0.41 & 0.68/0.53 & 0.30/0.25 & 0.29/0.25 & 0.46/0.47 & 0.90 & 0.90 \\ \hline
\text{Variant IV} & 0.41 & 0.56 & 0.66 & 0.47 & 0.57/0.49 & 0.53/0.53 & 0.59/0.49 & 0.48/0.42 & 0.58/0.46 & 0.45 & 0.63 \\ \hline
    \rowcolor{greyL}
\text{\alg} & 0.12 & 0.12 & 0.13 & 0.14 & 0.18/0.07 & 0.17/0.05 & 0.16/0.07 & 0.15/0.06 & 0.17/0.03 & 0.15 & 0.18 \\ \hline
\end{tabular}
\label{variant_result}
\end{table*}

\myparatight{Impact of degree of Non-IID}%
Fig.~\ref{delay_non_iid} evaluates the robustness of \alg under varying degrees of data heterogeneity by increasing the Non-IID parameter, which controls how unevenly class distributions are partitioned across clients. As the Non-IID level intensifies, client updates become more diverse and less representative of the global data distribution, significantly complicating the task of distinguishing benign updates from adversarial ones. Under mild heterogeneity, most methods maintain reasonable performance; however, as the Non-IID degree increases to more challenging levels (e.g., 0.7 and 0.9), several baseline defenses experience severe degradation. In particular, Zeno++ and AFLGuard become highly vulnerable even to relatively simple poisoning attacks, with testing error rates rising sharply. This failure can be attributed to their reliance on a trusted dataset at the server, whose distribution becomes increasingly misaligned with client data as heterogeneity grows, rendering similarity-based or reference-gradient checks ineffective.

In contrast, \alg demonstrates strong resilience across all evaluated Non-IID levels. Its testing error remains stable and consistently lower than those of competing methods, even under extreme heterogeneity. This robustness stems from the fact that \alg does not depend on any auxiliary trusted dataset or global reference update; instead, it leverages historical client update trajectories and Lipschitz-based smoothness constraints that naturally adapt to heterogeneous data distributions. As a result, benign but highly diverse client updates are preserved, while anomalous and malicious updates are effectively filtered out. These results indicate that \alg is particularly well-suited for real-world FL scenarios, where data heterogeneity is often severe and unavoidable, and further highlight its advantage over defenses that implicitly assume near-IID data distributions.

\myparatight{Impact of $\alpha$}%
In \alg, the parameter $\alpha$ determines the strictness of the Lipschitz-based filtering mechanism by specifying the percentile threshold used to classify received updates as benign or anomalous. A smaller $\alpha$ corresponds to a more conservative filter that rejects a larger portion of updates, while a larger $\alpha$ relaxes the filtering criterion and allows more updates to pass through. Table~\ref{alpha_impact} reports the performance of \alg under different choices of $\alpha$, demonstrating that the framework remains robust across a wide range of values. When $\alpha$ is set to lower values, \alg slightly increases its rejection rate, which can marginally affect convergence speed but provides strong protection against aggressive poisoning attacks. Conversely, higher $\alpha$ values admit more updates, improving learning efficiency while still maintaining strong robustness due to the subsequent estimation and Byzantine-robust aggregation steps.

Importantly, the results indicate that \alg is not overly sensitive to the precise tuning of $\alpha$. Across all evaluated values, \alg consistently achieves low testing error rates and, for targeted attacks, maintains very low attack success rates. This stability suggests that the Lipschitz-filtering mechanism effectively captures the normal smoothness patterns of benign client updates, even when the threshold is moderately relaxed or tightened. Moreover, the complementary design of \alg ensures that even if some malicious updates bypass the filter at higher $\alpha$ values, their influence is further mitigated by the update estimation process and coordinate-wise median aggregation. Overall, these findings show that \alg is robust to the choice of $\alpha$ and can be reliably deployed without requiring fine-grained parameter tuning, which is particularly advantageous in practical FL systems where attack characteristics are unknown in advance.

\myparatight{Different variants of \alg}%
In this section, we explore different variants of our proposed \alg.

\begin{list}{\labelitemi}{\leftmargin=1.15em \itemindent=-0.08em \itemsep=.1em}

\item {\bf Variant I:} The server uses the FedAvg rule~\cite{mcmahan2017communication} in Eq.~(\ref{agg_benign}) and Eq.~(\ref{agg_est}).

\item {\bf Variant II:} The server does not estimate updates from the remaining \( n-1 \) clients. It updates the global model using \(\bm{g}_i^{t-\tau_i}\) only if deemed benign; otherwise, no update is applied. 

\item {\bf Variant III:} The server does not discard malicious updates but continues estimating updates from the other \( n-1 \) clients. It then integrates both received and estimated updates using Eq.~(\ref{agg_benign}).

\item {\bf Variant IV:} Server disregards all received updates while still estimating updates from the remaining \( n-1 \) clients. It then combines these estimated updates by Eq.~(\ref{agg_est}).
\end{list}

Table~\ref{variant_result} compares the performance of \alg and its variants on the Fashion-MNIST dataset. The results emphasize the critical role of filtering malicious updates, incorporating estimated updates, and utilizing a robust aggregation strategy like Median. These components collectively enhance the effectiveness of \alg, demonstrating their superiority over alternative configurations.

\myparatight{Empirical validation of bounded estimation error}%
To empirically validate the bounded estimation error in Assumption~\ref{ass:est}, we measure the relative estimation error $\frac{\left\|\hat{{\bm{g}}}^t - {\bm{g}}^t\right\|_2}{\left\|{\bm{g}}^t\right\|_2}$ on the Fashion-MNIST dataset throughout training, where $\hat{{\bm{g}}}^t$ denotes the estimated update and ${\bm{g}}^t$ denotes the corresponding ground-truth update at round $t$. For each round interval, we compute the error at every round and report the average over the interval. As shown in Table~\ref{tab:estimation_error}, the relative error decreases steadily as training progresses and remains bounded, which empirically supports Assumption~\ref{ass:est} underlying our theoretical analysis.

\myparatight{Comparison with AsyncDefender~\cite{bai2025asyncdefender}}%
We further compare \alg with AsyncDefender~\cite{bai2025asyncdefender}. Table~\ref{tab:asyncdefender} in the Appendix reports the results of AsyncDefender under various attacks across five datasets. The symbol ``--'' for the Udacity dataset indicates that these attacks are not applicable, as Udacity is a regression task. 
Based on Tables~\ref{main_result}, \ref{main_result_Udacity}, and~\ref{tab:asyncdefender}, \alg consistently outperforms AsyncDefender across all evaluated settings.


\section{Limitations} 
\label{sec:discussion}

\myparatight{Privacy concern of \alg}%
\alg introduces a local update estimation mechanism that reconstructs missing client updates using historical global models and previously received gradients. While this estimation is performed entirely on the server side and does not require access to clients’ raw data, it nevertheless raises a distinct privacy consideration compared to standard asynchronous FL. Specifically, by approximating a client’s current update from its past behavior, the server implicitly infers information about how that client’s local objective evolves over time. Although this inferred update is only an approximation and is never shared externally, it may increase the amount of information the server can deduce about an individual client relative to schemes that strictly aggregate received updates only.

Importantly, this privacy concern is limited to the honest-but-curious server threat model and does not expose additional information to other clients or external adversaries. Moreover, \alg does not require storing raw data, labels, or intermediate activations, and the estimated updates are used solely for aggregation and discarded afterward. If stronger privacy guarantees are required, the estimation step can be combined with standard privacy-enhancing techniques such as update clipping, noise injection, or secure aggregation to limit potential information leakage. Therefore, while \alg slightly enlarges the inference capability of the server due to update estimation, it remains compatible with existing privacy-preserving mechanisms and maintains a practical balance between robustness and privacy in asynchronous FL.

\begin{table}[t]
\centering
\caption{Relative estimation error across training rounds on the Fashion-MNIST dataset.}
\label{tab:estimation_error}
\begin{tabular}{|c|c|}
\hline
Round & Relative estimation error\\
\hline
200--300         & 0.81 \\ \hline
300--400         & 0.72 \\ \hline
400--500         & 0.68 \\ \hline
500--600         & 0.65 \\ \hline
600--700         & 0.61 \\ \hline
700--Convergence & 0.59 \\
\hline
\end{tabular}
\end{table}

\begin{table}[t]
\centering
\caption{Running time (in seconds) under varying numbers of clients, where the CNN has 140,000 parameters.}
\label{tab:runtime_clients}
\begin{tabular}{|c|c|c|c|}
\hline
Method & 50 clients & 100 clients & 300 clients \\
\hline
AsyncSGD  & 0.02 & 0.02 & 0.04 \\ \hline
Zeno++    & 0.02 & 0.03 & 0.05 \\ \hline
AFLGuard  & 0.04 & 0.06 & 0.08 \\ \hline
\rowcolor{greyL}
\alg & 0.03 & 0.04 & 0.06 \\
\hline
\end{tabular}
\end{table}

\begin{table}[t]
\centering
\footnotesize
\addtolength{\tabcolsep}{-0.3pt}
\caption{Running time (in seconds) under varying CNN model sizes, where the number of clients is set to 50.}
\label{tab:runtime_modelsize}
\begin{tabular}{|c|c|c|c|}
\hline
Method & 140,000 parameters & 500,000 parameters & 1,000,000 parameters \\
\hline
AsyncSGD  & 0.02 & 0.03 & 0.04 \\ \hline
Zeno++    & 0.02 & 0.04 & 0.06 \\ \hline
AFLGuard  & 0.04 & 0.05 & 0.08 \\ \hline
\rowcolor{greyL}
\alg & 0.03 & 0.04 & 0.06 \\
\hline
\end{tabular}
\end{table}

\myparatight{Server's storage and computational expenses}%
Compared with standard asynchronous FL, \alg introduces additional but well-bounded server-side storage and computational costs. Let $n$ denote the total number of clients, $d$ the model dimension, and $\epsilon$ the L-BFGS buffer size. For storage, the server maintains a limited history of global model differences $\{\Delta \bm{w}_{t-\epsilon}, \ldots, \Delta \bm{w}_{t-1}\}$ and client-specific update differences $\{\Delta \bm{g}^{k}_{t-\epsilon}, \ldots, \Delta \bm{g}^{k}_{t-1}\}$ for each client $k$. Consequently, the total storage overhead is $O(n \epsilon d)$, which scales linearly with the number of clients and the model dimension. Since $\epsilon$ is a small constant in practice (e.g., $\epsilon = 3$ in our experiments), this overhead remains modest and does not increase with the number of training rounds.

In terms of computation, at each training round, the server performs three main operations. First, the Lipschitz-based filtering step computes vector norms and updates the percentile threshold, incurring a cost of $O(d)$ per received update. Second, the update estimation step applies L-BFGS-based Hessian--vector products for the $n-1$ non-uploading clients, resulting in a computational complexity of $O(n \epsilon d)$ per round. This procedure avoids explicit Hessian computation and relies only on vector inner products and linear combinations. Third, the coordinate-wise median aggregation over at most $n$ vectors incurs a cost of $O(n d)$, which is comparable to other Byzantine-robust aggregation rules. Overall, the per-round computational complexity of \alg is $O(n \epsilon d)$, which is higher than vanilla AsyncSGD but remains practical for realistic values of $n$, $d$, and $\epsilon$, representing a reasonable trade-off for improved robustness in adversarial asynchronous FL settings.

Table~\ref{tab:runtime_clients} and Table~\ref{tab:runtime_modelsize} report the running time (in seconds) of representative baselines and \alg on the Fashion-MNIST dataset. Specifically, Table~\ref{tab:runtime_clients} fixes the CNN architecture at 140,000 parameters while varying the number of clients, whereas Table~\ref{tab:runtime_modelsize} fixes the number of clients at 50 while varying the CNN model size. The results show that \alg incurs only modest overhead compared to AsyncSGD and scales gracefully with respect to both the number of clients and the model size. 
In terms of memory consumption, \alg maintains historical buffers and performs second-order approximations via L-BFGS, resulting in a server-side memory usage of 0.456 GB, which remains well within the capacity of standard hardware.


\section{Conclusion}

We present \alg, a Byzantine-resilient framework for asynchronous FL that effectively mitigates poisoning attacks. Our \alg systematically enhances model aggregation by first examining incoming local updates through a smoothness-based criterion, discarding updates that exhibit abrupt deviations from historical behavior. It then reconstructs missing client updates by exploiting prior model dynamics and historical update trajectories to compensate for partial participation. Finally, the server integrates both received and estimated updates using robust aggregation rules, such as the coordinate-wise median, to further limit the influence of adversarial updates. Extensive experimental results across multiple datasets and attack scenarios demonstrate the effectiveness and robustness of \alg in defending against poisoning attacks under fully asynchronous settings.

\begin{acks}
We thank the anonymous reviewers for their comments.
\end{acks}

\bibliographystyle{ACM-Reference-Format}

\bibliography{refs}

\appendices
\balance

\appendix

\begin{figure*}[!t]
	\centering
	\includegraphics[scale = 0.35]{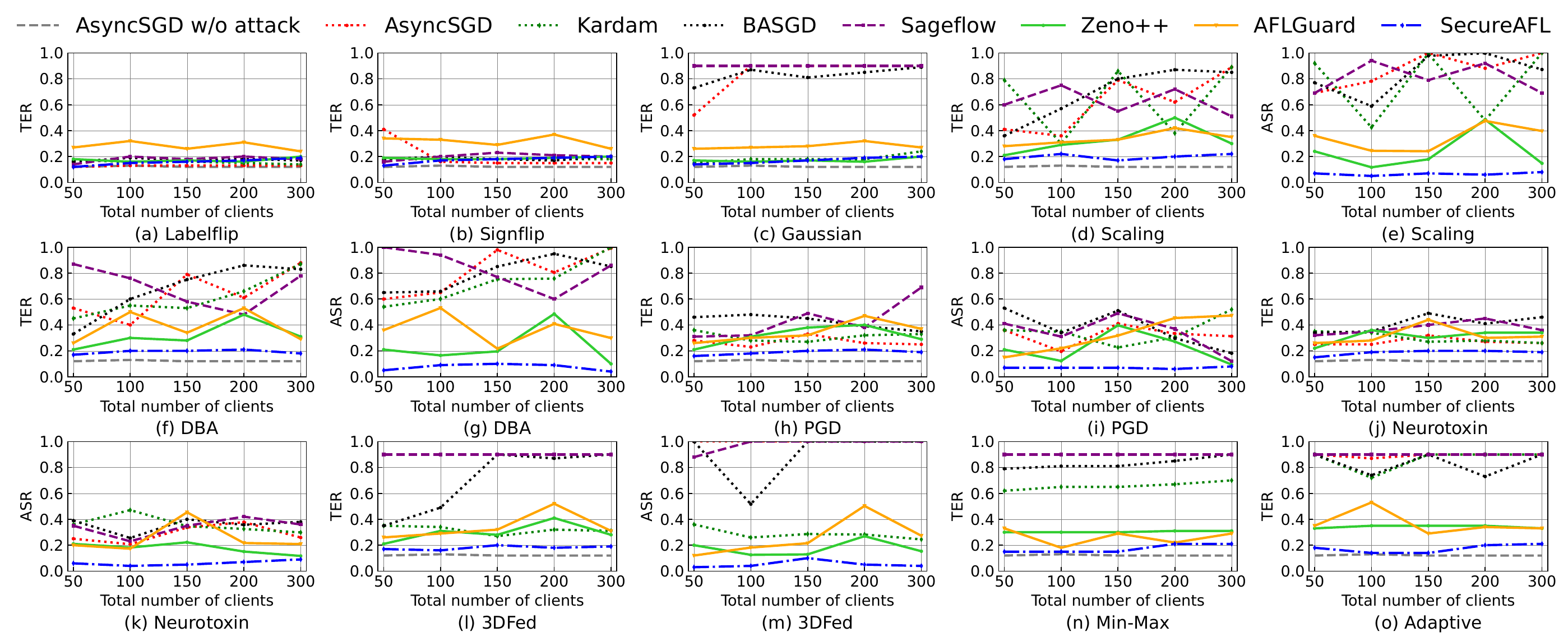}
	\caption{Impact of total number of clients on Fashion-MNIST dataset.}
	\label{number_clients}
\end{figure*}

\begin{table*}[h]
\caption{Performance of AsyncDefender across different datasets and attack types.}
\centering
\addtolength{\tabcolsep}{-1.8pt}
\begin{tabular}{|c|c|c|c|c|c|c|c|c|c|c|c|}
\hline
\text{Dataset} & \text{No attack} & \text{Labelflip} & \text{Signflip} & \text{Gaussian} & \text{Scaling} & \text{DBA} & \text{PGD} & \text{Neurotoxin} & \text{3DFed} & \text{Min-Max} & \text{Adaptive} \\ \hline
\text{Fashion-MNIST}    & 0.24 & 0.25 & 0.33 & 0.31 & 0.38/0.44 & 0.34/0.43 & 0.33/0.29 & 0.32/0.27 & 0.34/0.31 & 0.39 & 0.43 \\ \hline
\text{CIFAR-10}         & 0.26 & 0.40 & 0.42 & 0.49 & 0.58/0.37 & 0.53/0.42 & 0.55/0.44 & 0.52/0.40 & 0.57/0.43 & 0.55 & 0.66 \\ \hline
\text{CIFAR-100}        & 0.50 & 0.58 & 0.57 & 0.66 & 0.73/0.53 & 0.74/0.55 & 0.71/0.67 & 0.73/0.44 & 0.72/0.57 & 0.93 & 0.89 \\ \hline
\text{Tiny-ImageNet}    & 0.63 & 0.71 & 0.75 & 0.71 & 0.70/0.15 & 0.72/0.18 & 0.69/0.12 & 0.70/0.15 & 0.71/0.15 & 0.73 & 0.72 \\ \hline
\text{Udacity}          & 0.18 & --   & 0.41 & 0.32 & --        & --        & --        & --        & --        & 0.29 & 0.33 \\ \hline
\end{tabular}
\label{tab:asyncdefender}
\end{table*}

\section{Important Lemmas}
We first derive a tracking error model of the form
\begin{align}
\bm{g}^t = \nabla F_{\mathcal H}(\bm{w}^t) + \bm{\delta}^t,\quad
\mathbb{E}\bigl[\|\bm{\delta}^t\|^2\bigr]\le E_{\mathrm{track}}^2. 
\end{align}
This is obtained by combining staleness, estimation error, heterogeneity, and median robustness under Byzantine contamination.

\begin{lem}
\label{lem:stale}
Recall that the server model is updated as
\(
\bm{w}^{s+1} = \bm{w}^s - \eta \bm{g}^s,
\)
where $\bm{g}^s$ denotes the aggregated update applied by the server at round $s$. Under Assumptions~\ref{ass:smooth} and~\ref{ass:stale}, for any benign client $j\in\mathcal{H}$,

\begin{align}
\|\nabla f_{j}(\bm{w}^{t-\tau_i}) - \nabla f_{j}(\bm{w}^t)\|
\le L \|\bm{w}^{t-\tau_i}-\bm{w}^t\|
\le L\eta \sum_{s=t-\tau_{\max}}^{t-1}\|\bm{g}^s\|. \label{eq:stale_bound}
\end{align}
Consequently, using Assumption~\ref{ass:filter_bounded} and Jensen's inequality,
\begin{equation}
\begin{aligned}
\mathbb{E}\bigl[\|\nabla f_{j}(\bm{w}^{t-\tau_i}) - \nabla f_{j}(\bm{w}^t)\|^2\bigr]
&\le L^2 \eta^2 \tau_{\max} \sum_{s=t-\tau_{\max}}^{t-1} \mathbb{E}\|\bm{g}^s\|^2\\
&\le L^2 \eta^2 \tau_{\max}^2 G^2. \label{eq:stale_second_moment}
\end{aligned}
\end{equation}

\end{lem}
\begin{proof}
The first inequality follows from $L$-smoothness. For the second, use telescoping:
\begin{align}
\bm{w}^{t}-\bm{w}^{t-\tau_i}
=
\sum_{s=t-\tau_i}^{t-1}(\bm{w}^{s+1}-\bm{w}^s)
=
-\eta\sum_{s=t-\tau_i}^{t-1}\bm{g}^s,
\end{align}
hence
\begin{align}
\|\bm{w}^{t}-\bm{w}^{t-\tau_i}\|
\le \eta\sum_{s=t-\tau_i}^{t-1}\|\bm{g}^s\|
\le \eta\sum_{s=t-\tau_{\max}}^{t-1}\|\bm{g}^s\|.
\end{align}
Squaring and using $(\sum_{i=1}^k a_i)^2\le k\sum_{i=1}^k a_i^2$ yields the first inequality in Eq.~(\ref{eq:stale_second_moment}). The last inequality uses Eq.~(\ref{eq:input_norm_bound}).
\end{proof}
At round $t$, define for each benign client $j\in\mathcal{H}$ the contemporaneous gradient
\begin{align}
\bm{h}_j^t = \nabla f_j(\bm{w}^t). 
\end{align}

\begin{align}
\bm{u}_j^t =
\begin{cases}
\bm{g}_{i}^{t-\tau_i}, & \text{if } j=i \text{ and the filter accepts},\\
\hat{\bm{g}}_j^t, & \text{if } j\in\mathcal{S}.
\end{cases}
\label{eq:def_ukt}
\end{align}
\begin{lem}
\label{lem:benign_vector_error}
Under Assumptions~\ref{ass:benign_upload},~\ref{ass:est}, and Lemma~\ref{lem:stale}, for any round $t$,
\begin{equation}
\begin{aligned}
  &\mathbb{E}\bigl[\|\bm{u}_j^t - \bm{h}_j^t\|^2 \mid \mathcal{F}_t\bigr]\\
&\le
\begin{cases}
2\sigma^2 + 2L^2\eta^2\tau_{\max}^2 G^2, & \text{if } j=i\in\mathcal{H} \text{ and accepted},\\[3pt]
\varepsilon_{\mathrm{est}}^2, & \text{if } j\in\mathcal{H}\cap\mathcal{S}.
\end{cases}
\label{eq:benign_contrib_bound}    
\end{aligned}
\end{equation}

\end{lem}
\begin{proof}
If $j\in\mathcal{H}\cap\mathcal{S}$, then $\bm{u}_j^t=\hat{\bm{g}}_j^t$ and Eq.~(\ref{eq:benign_contrib_bound}) follows from Assumption~\ref{ass:est}.
If $j=i\in\mathcal{H}$ and accepted, then
\begin{align}
\bm{u}_j^t - \bm{h}_j^t = \bigl(\nabla f_j(\bm{w}^{t-\tau_i})-\nabla f_j(\bm{w}^t)\bigr) + \bm{\xi}_{i,t}.
\end{align}

Using $\|a+b\|^2\le 2\|a\|^2+2\|b\|^2$  taking conditional expectation, and applying Lemma~\ref{lem:stale} and Assumption~\ref{ass:benign_upload},
\begin{equation}
\begin{aligned}
\mathbb{E}\bigl[\|\bm{u}_j^t-\bm{h}_j^t\|^2\mid \mathcal{F}_t\bigr]
\le
&2\mathbb{E}\bigl[\|\nabla f_j(\bm{w}^{t-\tau_i})-\nabla f_j(\bm{w}^t)\|^2\mid \mathcal{F}_t\bigr]
\\
&+
2\mathbb{E}\bigl[\|\bm{\xi}_{i,t}\|^2\mid \mathcal{F}_t\bigr].
\end{aligned}
\end{equation}
\end{proof}

\begin{lem}[Median aggregation as a bounded tracking error]
\label{lem:median_tracking}
Suppose Assumptions~\ref{ass:byz_fraction},~\ref{ass:filter_bounded}, and~\ref{ass:heterogeneity} hold.
Let
\begin{align}
\bm{g}^t = \mathrm{Median}(\mathcal{V}_t).
\end{align}
For every client $j\in[n]$, the possibly estimated vector that is actually fed into the median at round $t$ as
\begin{align}
\bm{u}_j^t \in \mathcal{V}_t,\quad \mathcal{V}_t=\{\bm{u}_j^t\}_{j\in[n]},
\end{align}
so that at most $b$ elements of $\mathcal{V}_t$ are Byzantine and the remaining correspond to benign clients in $\mathcal{H}$.
Let
\begin{align}
\bar{\bm{h}}^t = \nabla F_{\mathcal{H}}(\bm{w}^t)=\frac{1}{m}\sum_{j\in\mathcal{H}}\nabla f_j(\bm{w}^t),
\quad
\bm{h}_j^t = \nabla f_j(\bm{w}^t).
\end{align}
Assume additionally the following median second-moment robustness property holds for the coordinate-wise median: there exists a constant $C_{\mathrm{med}}>0$ (depending only on the Byzantine fraction bound) such that, for all $t$,
\begin{align}
\mathbb{E}\bigl[\|\mathrm{Median}(\mathcal{V}_t)-\bar{\bm{h}}^t\|^2 \mid \mathcal{F}_t\bigr]
\le
\frac{C_{\mathrm{med}}}{m}\sum_{j\in\mathcal{H}}\mathbb{E}\bigl[\|\bm{u}_j^t-\bar{\bm{h}}^t\|^2\mid \mathcal{F}_t\bigr].
\label{eq:median_variance_bound}
\end{align}
Then there exists a random vector $\bm{\delta}^t$ such that
\begin{align}
\bm{g}^t = \nabla F_{\mathcal{H}}(\bm{w}^t) + \bm{\delta}^t,\quad
\mathbb{E}\bigl[\|\bm{\delta}^t\|^2\bigr]\le E_{\mathrm{track}}^2, \label{eq:tracking_form}
\end{align}
where
\begin{align}
E_{\mathrm{track}}^2
\le
C_{\mathrm{med}}\Bigl(
\zeta^2
+\varepsilon_{\mathrm{est}}^2
+\sigma^2
+L^2\eta^2\tau_{\max}^2 G^2
\Bigr).
\label{eq:Etrack_explicit}
\end{align}
\end{lem}

\begin{proof}
Fix $t$ and condition on $\mathcal{F}_t$. Let $\bar{\bm{h}}^t = \nabla F_{\mathcal{H}}(\bm{w}^t)$.
Define the tracking error $\bm{\delta}^t = \bm{g}^t-\bar{\bm{h}}^t$, so that $\bm{g}^t=\bar{\bm{h}}^t+\bm{\delta}^t$.
By the assumed median second-moment robustness property Eq.~\eqref{eq:median_variance_bound},
\begin{equation}
\begin{aligned}
\mathbb{E}\bigl[\|\bm{\delta}^t\|^2 \mid \mathcal{F}_t\bigr]
&=
\mathbb{E}\bigl[\|\bm{g}^t-\bar{\bm{h}}^t\|^2 \mid \mathcal{F}_t\bigr]\\
&\le
\frac{C_{\mathrm{med}}}{m}\sum_{j\in\mathcal{H}}\mathbb{E}\bigl[\|\bm{u}_j^t-\bar{\bm{h}}^t\|^2\mid \mathcal{F}_t\bigr].
\label{eq:med_step1}
\end{aligned}
\end{equation}

For each benign client $j\in\mathcal{H}$, add and subtract $\bm{h}_j^t$ and use $\|a+b\|^2\le 2\|a\|^2+2\|b\|^2$:
\begin{equation}
\begin{aligned}
\|\bm{u}_j^t-\bar{\bm{h}}^t\|^2
&=
\|\bm{u}_j^t-\bm{h}_j^t+\bm{h}_j^t-\bar{\bm{h}}^t\|^2\\
&\le
2\|\bm{u}_j^t-\bm{h}_j^t\|^2
+
2\|\bm{h}_j^t-\bar{\bm{h}}^t\|^2.
\label{eq:decomp_uk}
\end{aligned}
\end{equation}

Substituting Eq.~\eqref{eq:decomp_uk} into Eq.~\eqref{eq:med_step1} yields
\begin{equation}
\begin{aligned}
\mathbb{E}\bigl[\|\bm{\delta}^t\|^2 \mid \mathcal{F}_t\bigr]
\le&
\frac{2C_{\mathrm{med}}}{m}\sum_{j\in\mathcal{H}}
\mathbb{E}\bigl[\|\bm{u}_j^t-\bm{h}_j^t\|^2\mid \mathcal{F}_t\bigr]\\
&+
\frac{2C_{\mathrm{med}}}{m}\sum_{j\in\mathcal{H}}
\|\bm{h}_j^t-\bar{\bm{h}}^t\|^2.
\label{eq:med_step2}
\end{aligned}
\end{equation}

By Assumption~\ref{ass:heterogeneity}, the second term satisfies
\begin{align}
\frac{1}{m}\sum_{j\in\mathcal{H}}
\|\bm{h}_j^t-\bar{\bm{h}}^t\|^2
\le
\zeta^2.
\label{eq:hetero_use}
\end{align}
For the first term, Lemma~\ref{lem:benign_vector_error} gives, for each benign $j$ at round $t$,
\begin{align}
\mathbb{E}\bigl[\|\bm{u}_j^t-\bm{h}_j^t\|^2\mid \mathcal{F}_t\bigr]
\le
\varepsilon_{\mathrm{est}}^2
\quad\text{or}\quad
2\sigma^2+2L^2\eta^2\tau_{\max}^2G^2,
\end{align}
depending on whether $\bm{u}_j^t$ is an estimator output or an accepted stale upload.
Therefore, in all cases,
\begin{align}
\frac{1}{m}\sum_{j\in\mathcal{H}}
\mathbb{E}\bigl[\|\bm{u}_j^t-\bm{h}_j^t\|^2\mid \mathcal{F}_t\bigr]
\le
\varepsilon_{\mathrm{est}}^2
+
2\sigma^2
+
2L^2\eta^2\tau_{\max}^2G^2.
\label{eq:surrogate_avg_bound}
\end{align}
Combining Eq.~\eqref{eq:med_step2}, Eq.~\eqref{eq:hetero_use}, and Eq.~\eqref{eq:surrogate_avg_bound} and absorbing constant factors into $C_{\mathrm{med}}$ yields
\begin{align}
\mathbb{E}\bigl[\|\bm{\delta}^t\|^2 \mid \mathcal{F}_t\bigr]
\le
C_{\mathrm{med}}\Bigl(
\zeta^2+\varepsilon_{\mathrm{est}}^2+\sigma^2+L^2\eta^2\tau_{\max}^2G^2
\Bigr).
\end{align}
Taking total expectation over $\mathcal{F}_t$ proves Eq.~\eqref{eq:tracking_form}--\eqref{eq:Etrack_explicit}.
\end{proof}

\section{Proofs for Theorem ~\ref{thm:main}} \label{sec:appendix_main}
\begin{proof}
Now, with all assumptions and the important lemmas we have mentioned above. We now illustrate the detailed proof for theorem \ref{thm:main}.
First, By $L$-smoothness of $F_{\mathcal H}$ (Assumption~\ref{ass:smooth}), for the update $\bm{w}^{t+1}=\bm{w}^t-\eta \bm{g}^t$, we have,
\begin{align}
F_{\mathcal H}(\bm{w}^{t+1})
\le
F_{\mathcal H}(\bm{w}^t)
-\eta\langle \nabla F_{\mathcal H}(\bm{w}^t),\bm{g}^t\rangle
+\frac{L\eta^2}{2}\|\bm{g}^t\|^2.
\label{eq:descent}
\end{align}
Then, by Lemma~\ref{lem:median_tracking}, $\bm{g}^t=\nabla F_{\mathcal H}(\bm{w}^t)+\bm{\delta}^t$ with $\mathbb{E}\|\bm{\delta}^t\|^2\le E_{\mathrm{track}}^2$. We get,
\begin{align}
-\langle \nabla F_{\mathcal H}(\bm{w}^t),\bm{g}^t\rangle
=
-\|\nabla F_{\mathcal H}(\bm{w}^t)\|^2-\langle \nabla F_{\mathcal H}(\bm{w}^t),\bm{\delta}^t\rangle.
\end{align}
By Cauchy--Schwarz and Young's inequality, $\langle a,b\rangle \le \frac{1}{2}\|a\|^2+\frac{1}{2}\|b\|^2$, hence
\begin{align}
-\langle \nabla F_{\mathcal H}(\bm{w}^t),\bm{g}^t\rangle
\le
-\frac{1}{2}\|\nabla F_{\mathcal H}(\bm{w}^t)\|^2+\frac{1}{2}\|\bm{\delta}^t\|^2.
\label{eq:inner_bound}
\end{align}
Therefore, we have
\begin{align}
\|\bm{g}^t\|^2 = \|\nabla F_{\mathcal H}(\bm{w}^t)+\bm{\delta}^t\|^2
\le 2\|\nabla F_{\mathcal H}(\bm{w}^t)\|^2+2\|\bm{\delta}^t\|^2.
\end{align}
Substitute this and Eq.~(\ref{eq:inner_bound}) into Eq.~(\ref{eq:descent}):
\begin{equation}
\begin{aligned}
F_{\mathcal H}(\bm{w}^{t+1})
\le&
F_{\mathcal H}(\bm{w}^t)
+\eta\Bigl(-\tfrac{1}{2}\|\nabla F_{\mathcal H}(\bm{w}^t)\|^2+\tfrac{1}{2}\|\bm{\delta}^t\|^2\Bigr)\\
&+\frac{L\eta^2}{2}\Bigl(2\|\nabla F_{\mathcal H}(\bm{w}^t)\|^2+2\|\bm{\delta}^t\|^2\Bigr)\\
=&
F_{\mathcal H}(\bm{w}^t)
-\eta\Bigl(\tfrac{1}{2}-L\eta\Bigr)\|\nabla F_{\mathcal H}(\bm{w}^t)\|^2
+\eta\Bigl(\tfrac{1}{2}+L\eta\Bigr)\|\bm{\delta}^t\|^2.
\end{aligned}
\end{equation}

Under $\eta\le \frac{1}{4L}$, we have $\frac{1}{2}-L\eta \ge \frac{1}{4}$ and $\frac{1}{2}+L\eta \le 1$, hence
\begin{align}
F_{\mathcal H}(\bm{w}^{t+1})\le F_{\mathcal H}(\bm{w}^t)-\frac{\eta}{4}\|\nabla F_{\mathcal H}(\bm{w}^t)\|^2 + \eta\|\bm{\delta}^t\|^2.
\label{eq:descent_simple}
\end{align}
Take expectations and apply $\mathbb{E}\|\bm{\delta}^t\|^2\le E_{\mathrm{track}}^2$:
\begin{align}
\mathbb{E}F_{\mathcal H}(\bm{w}^{t+1})
\le
\mathbb{E}F_{\mathcal H}(\bm{w}^t)-\frac{\eta}{4}\mathbb{E}\|\nabla F_{\mathcal H}(\bm{w}^t)\|^2+\eta E_{\mathrm{track}}^2.
\end{align}
Rearrange, sum $t=0$ to $T-1$, and use $F_{\mathcal H}(\bm{w}^T)\ge F_{\mathcal H}{\star}$:
\begin{align}
\frac{\eta}{4}\sum_{t=0}^{T-1}\mathbb{E}\|\nabla F_{\mathcal H}(\bm{w}^t)\|^2
\le
F_{\mathcal H}(\bm{w}^0)-F_{\mathcal H}{\star} + T\eta E_{\mathrm{track}}^2.
\end{align}
Divide both sides by $\eta T/4$ to obtain Eq.~(\ref{eq:main_rate}).
\end{proof}

\section{Details of Datasets}
\label{app_dataset}

\myparatight{Fashion-MNIST~\cite{xiao2017online}}%
Fashion-MNIST, derived from Zalando's article images, comprises 70,000 grayscale images with a resolution of 28×28 pixels, categorized into 10 distinct classes. The dataset is divided into 60,000 images for training and 10,000 for testing.

\myparatight{CIFAR-10~\cite{cifar10data}}%
CIFAR-10 consists of 60,000 colored images, each with a resolution of 32×32 pixels, categorized into 10 distinct classes. The dataset is structured into 50,000 images for training and 10,000 for testing.

\myparatight{CIFAR-100~\cite{cifar10data}}%
CIFAR-100 follows the same format as CIFAR-10 but features a more fine-grained classification, comprising 60,000 color images of 32×32 pixels. These images are distributed across 100 distinct classes, which are further organized into 20 broader categories. The dataset includes 50,000 images for training and 10,000 for testing.

\myparatight{Tiny-ImageNet~\cite{deng2009imagenet}}%
Tiny-ImageNet is a compact variant of the ImageNet dataset tailored for large-scale image recognition. It comprises 110,000 images spanning 200 classes, with 100,000 designated for training and 10,000 set aside for testing.

\myparatight{Udacity~\cite{Udacity}}%
The Udacity dataset is a dataset for regression tasks.
It supports autonomous driving research by enabling the prediction of a vehicle’s steering angle within a simulated environment provided by Udacity. It comprises images recorded from the onboard camera during human-driven demonstrations. Leveraging this data, a model is trained to infer steering angles, with its performance ultimately assessed on an unseen test track.

\section{Details of Poisoning Attacks}
\label{app_attack}

\myparatight{Label flipping (Labelflip) attack~\cite{tolpegin2020data}}%
The label flipping attack alters the training labels of malicious clients by transforming each label \( y \) into \( z - 1 - y \), where \( z \) represents the total number of classes.

\myparatight{Signflip attack~\cite{fang2020local}}%
The Signflip attack disrupts model updates by having malicious clients invert the sign of every element in their update vector. This is accomplished by multiplying the entire vector by -1 before submission.

\myparatight{Gaussian attack~\cite{blanchard2017machine}}%
The Gaussian attack involves malicious clients fabricating model updates by sampling from a Gaussian distribution with a mean of zero and a standard deviation of 200.

\myparatight{Scaling attack~\cite{bagdasaryan2020backdoor}}%
In a Scaling attack, the attacker embeds distinct trigger patterns into a portion of the training data belonging to malicious clients, ensuring these triggers correspond to a predefined target label. Additionally, malicious clients amplify their local model updates before transmitting them to the server.

\myparatight{DBA attack~\cite{xie2019dba}}%
The DBA attack takes advantage of the decentralized structure of FL by fragmenting a global trigger pattern into unique local patterns. These patterns are then systematically injected into the training data of malicious clients.

\myparatight{Projected gradient descent (PGD) attack~\cite{sun2019can}}%
The attacker can employ projected gradient descent to train the backdoor model, ensuring that in each training round, the model is updated and then constrained within an \(\ell_2\) ball centered around the previous round’s model.

\myparatight{Neurotoxin attack~\cite{zhang2022neurotoxin}}%
The Neurotoxin attack strengthens backdoor persistence in FL models by targeting parameters that change minimally during training, ensuring the backdoor remains effective despite continual updates.

\myparatight{3DFed attack~\cite{li20233dfed}}%
3DFed attack is a stealthy, multi-layered framework for backdoor attacks in black-box FL systems. It employs constrained-loss training, noise masking, and a decoy model to evade detection.

\myparatight{Min-Max attack~\cite{shejwalkar2021manipulating}}%
Min-Max is an untargeted attack independent of aggregation rules, where the attacker stealthily manipulates updates from malicious clients.

\myparatight{Adaptive attack~\cite{shejwalkar2021manipulating}}%
In the Adaptive attack, the attacker, aware of the server's \alg defense, manipulates malicious client updates to maximize deviation from the pre-attack aggregated model.

\section{Details of Compared Methods}
\label{app_method}

\myparatight{AsyncSGD~\cite{zheng2017asynchronous}}%
It updates the global model immediately whenever a client submits a model update to the server.

\myparatight{Kardam~\cite{damaskinos2018asynchronous}}%
Kardam considers an update as malicious if it exhibits a substantial deviation from past updates.

\myparatight{BASGD~\cite{yang2021basgd}}%
BASGD organizes client updates into buffers based on a mapping table. Once filled, the server averages each buffer, computes the median of these averages, and updates the global model accordingly.

\myparatight{Sageflow~\cite{park2021sageflow}}%
Sageflow uses a trusted dataset to assess the quality of client updates. It mitigates the impact of stragglers through staleness-aware grouping and enhances robustness against adversarial attacks using entropy-based filtering and loss-weighted averaging.

\myparatight{Zeno++~\cite{xie2020zeno++}}%
Zeno++ verifies client updates using a trusted dataset. The server computes its own update and assesses its alignment with the client's update via cosine similarity. If the similarity is positive, the client's update is rescaled before being applied to the global model.

\myparatight{AFLGuard~\cite{fang2022aflguard}}%
AFLGuard also relies on a trusted dataset to generate a reference update. A received model is deemed benign if it aligns positively with this reference.

\end{document}